\documentclass[journal,12pt,onecolumn,draftclsnofoot,]{IEEEtran}
\usepackage{graphicx}
\usepackage{float}
\usepackage[center]{caption}
\usepackage{subcaption}
\usepackage{tikz}
\usepackage{cite}

\usepackage{titlecaps, stringstrings}
\usepackage{mathtools}
\usepackage{amsfonts}
\usepackage{amsmath}
\usepackage[nocomma]{optidef}
\usepackage[english]{babel}
\usepackage[utf8]{inputenc}
\usepackage{soul,color}
\usepackage{stfloats}

\DeclarePairedDelimiter\floor{\lfloor}{\rfloor}

\usepackage[english]{babel}

\usepackage{blindtext}
\usepackage{amssymb}

\usepackage{tabu}
\usepackage{xcolor}
\usepackage{lettrine}

\usepackage{algorithm,algorithmic}

\usepackage{amsthm}
\usepackage{tabularx,booktabs, ragged2e}
\newcolumntype{C}{>{\centering\arraybackslash}X} 
\usepackage{lipsum}
\newtheorem{theorem}{Theorem}

\newtheorem{lemma}[theorem]{Lemma}

\newtheorem{definition}{Definition}

\newtheorem{assumption}{Assumption}
\newtheorem{remark}{Remark}

\usepackage{mathtools,amssymb,lipsum}
\IEEEoverridecommandlockouts
\title{Intelligent Resource Slicing for eMBB and URLLC Coexistence in 5G and Beyond: A Deep Reinforcement Learning Based Approach}
\author{Madyan Alsenwi\thanks{Madyan Alsenwi, S. R. Pandey, and C. S. Hong are with the Department of Computer Science and Engineering, Kyung Hee University, Yongin 17104, South Korea (email: \{malsenwi, shashiraj, anupam, cshong\}@khu.ac.kr).}, Nguyen H. Tran\thanks{ N. H. Tran is with the School of Computer Science, University of Sydney, NSW 2006, Australia (e-mail: nguyen.tran@sydney.edu.au).}, \IEEEmembership{Senior Member, IEEE}, Mehdi Bennis\thanks{M. Bennis is with the Department of Communications Engineering, University of Oulu, FI-90014 Oulu, Finland, and also with the Department of Computer Science and Engineering, Kyung Hee University, Yongin 17104, South Korea (e-mail: mehdi.bennis@oulu.fi)}, \IEEEmembership{Senior Member, IEEE}, Shashi Raj Pandey, Anupam Kumar Bairagi\thanks{A. K. Bairagi is with the department of Computer Science and Engineering, Khulna University, Khulna 9208, Bangladesh, and also with the Department of Computer Science and Engineering, Kyung Hee University, Yongin 17104, South Korea (e-mail: anupam@khu.ac.kr).},  \IEEEmembership{Member, IEEE},  and Choong Seon Hong, \IEEEmembership{Senior Member, IEEE}} 

\begin{document}

	\maketitle
% 	\vspace{-1.4cm}
	\begin{abstract}
% 	\vspace{-0.3cm}
	In this paper, we study the resource slicing problem in a dynamic multiplexing scenario of two distinct 5G services, namely Ultra-Reliable Low Latency Communications (URLLC) and enhanced Mobile BroadBand (eMBB). While eMBB services focus on high data rates, URLLC is very strict in terms of latency and reliability. In view of this, the resource slicing problem is formulated as an optimization problem that aims at maximizing the eMBB data rate subject to a URLLC reliability constraint, while considering the variance of the eMBB data rate to reduce the impact of immediately scheduled URLLC traffic on the eMBB reliability. To solve the formulated problem, an optimization-aided Deep Reinforcement Learning (DRL) based framework is proposed, including: \emph{1) eMBB resource allocation phase}, and \emph{2) URLLC scheduling phase}. In the first phase, the optimization problem is decomposed into three subproblems and then each subproblem is transformed into a convex form to obtain an approximate resource allocation solution. In the second phase, a DRL-based algorithm is proposed to intelligently distribute the incoming URLLC traffic among eMBB users. Simulation results show that our proposed approach can satisfy the stringent URLLC reliability while keeping the eMBB reliability higher than $90\%$.   
	\end{abstract}
% 	\vspace{-0.7cm}
	\begin{IEEEkeywords}
		5G NR, resource slicing, eMBB, URLLC, risk-sensitive, deep reinforcement learning. 
% 			\vspace{-0.4cm}
	\end{IEEEkeywords}
% 	\vspace{-0.5cm}
	\section{Introduction}
% 		\vspace{-0.4cm}
	\lettrine{T}{he} services supported by the 5th Generation (5G) New Radio (NR) fall under three categories, i.e., enhanced Mobile Broad Band (eMBB), massive Machine-Type Communications (mMTC), and Ultra-Reliable Low-Latency Communications (URLLC). eMBB is designed to accommodate high data rate applications such as 4K video and Virtual Reality (VR). Specifically, eMBB service can be considered as an extension of LTE-Advanced broadband service which allows higher data rate and coding over large transmission blocks for a long time interval. Therefore, the objective of eMBB service is to achieve high data rate while satisfying a moderate reliability with packet error rate (PER) of $10^{-3}$ \cite{Bennis2018, Popovski_Access}. On the contrary, mMTC aims at serving a large number of Internet of Things (IoT) devices sending data sporadically with a low and fixed uplink transmission rate. A large number of mMTC devices may connect to a Base Station (BS) making it infeasible to allocate a priori resource to each device. Generally, mMTC devices, such as sensing,  metering, and monitoring, focus on energy-efficiency \cite{Park_Wireless_Communications}.
	
	%	 Recently, different technologies were proposed to support mMTC such as LoRa and SigFox in unlicensed bands and narrowband IoT (NB-IoT) in licensed bands. These technologies provide mMTC connectivity with low operation cost and low power consumption. However, its limitations come when the number of devices outnumbers the available resources for transmission. 
	
	Meanwhile, URLLC services target mission critical communications such as autonomous vehicles, tactile internet, or remote surgery. In general, URLLC transmissions are sporadic with a short packet size and with relatively low data rate. The main requirements of URLLC transmission are ultra-high reliability with  a PER around $10^{-5}$ and low latency. Due to its low latency requirement, URLLC transmissions are localized in time with short Transmission Time Intervals (sTTI). In 4G systems, the control signaling takes a large portion of the transmission latency, i.e., $0.3-0.4$ ms. Thus, designing a short packet transmission system with latency of $0.5$ ms may cause wasting of more than $60\%$ of resources for control overheads. To this end, many changes on the physical layer design have been introduced in 5G NR systems in order to support URLLC services \cite{Popovski_Access,Park_Wireless_Communications}.     
% 	\vspace{-0.7cm}
	\subsection{Physical layer enablers for URLLC in 5G NR}
% 		\vspace{-0.5cm}
	We discuss the 5G NR to support both defined services, i.e., eMMB, and URLLC. Generally, 5G NR supports multiple waveform configurations (numerology) and thus radio frame gets different shapes. The sub-carrier spacing of the low band outdoor macro networks is $15$ kHz while it is $30$ kHz in outdoor small cell networks. However, the higher frequency bands come with higher sub-carrier spacing, i.e., the sub-carrier spaces of $60$ kHz and $120$ kHz are chosen for the $5$ GHz unlicensed bands and the $28$ GHz mmWave bands, respectively. \cite{3gpp_release15}. In time domain, the length of a radio frame and a sub-frame are always, regardless of numerology, $10$ ms and $1$ ms, respectively. The difference is the number of time slots within a sub-frame and the number of symbols within a time slot\footnote{The number of symbols is fixed for all numerology and it only changes slot configuration type, i.e., for the slot configuration ``0", the number of symbols for a time slot is always $14$ while it is $7$ for slot configuration ``1".}. Hence, a Resource Block\footnote{A RB is defines as a group of OFDM sub-carriers for a time slot duration which is the smallest frequency-time unit that can be assigned to a user.} (RB) has different structures depending on the numerology.
	% Fig. \ref{5g_numerlogy} shows the 5G NR numerology of $15$ KHz and $60$ KHz sub-carrier spacing.
	%\begin{figure}[t]
	%	\centering
	%	\includegraphics[width=3.5in, height=2.5in]{Figs/5G_Numerlogy}
	%	%    \includegraphics[width=7in]{Figs/5G_Numerlogy.pdf}
	%	\caption{5G NR frame structure of $15$ KHz and $60$ KHz sub-carrier spacing numerology.}
	%	\label{5g_numerlogy}
	%\end{figure}
	
	To support low latency transmission of URLLC, one option is to reduce the symbol period by controlling the sub-carrier spacing, i.e., the symbol length can be reduced to half by doubling the sub-carrier spacing. This is relevant in mmWave bands (above $6$ GHz) as the cell radius is smaller due to the path loss inducing smaller channel delay spread compared to the conventional cellular systems. However, this approach cannot be applied to bands lower than $6$ GHz due to the large delay spread. Another option is to reduce the number of symbols in the packet TTI, i.e., using mini-slot (short TTI) level of 2-3 symbols and slot level (e.g. 7 symbols) transmissions. In summary, we can achieve a TTI smaller than $1$ ms by adjusting the number of symbols and the symbol period. Going further, to bring in more efficiency and reduce latency, a concept called Code Block Groups (CBGs) based transmission is proposed in 5G NR which divides the large transport block into smaller Code Blocks (CBs). Furthermore, the smaller CBs are further grouped into CBGs. Here, users decode CBGs and send feedback (ACK/NACK) for each individual group.
	
	We exploit the aforementioned facts to design an efficient mechanism to tackle the coexistence problem of eMBB and URLLC services. In particular, we leverage the frame structure flexibility of 5G NR to design a resource allocation framework to satisfy the specific  requirements of each service.        
% 	\vspace{-0.7cm}
	\subsection{Motivation}
% 	\vspace{-0.5cm}
	The coexistence of these heterogeneous services with distinct requirements mandates an efficient resource slicing framework that can satisfy the requirements of each service. Specifically, the incoming URLLC packets during the scheduling period of eMBB transmissions cannot be delayed due to its strict latency requirement. To this end, two approaches have been adopted in the third Generation Partnership Project (3GPP) standard \cite{Popovski_Access, 3gpp_release15}:
	\begin{itemize}
		\item \textbf{Preemptive (Puncturing) scheduling:} URLLC traffic will be scheduled in short TTIs on top of the ongoing eMBB transmissions. In other words, the Next Generation NodeB (gNB) stops eMBB transmission during short TTIs of URLLC transmission to ensure URLLC latency. This protocol is efficient in terms of reducing the URLLC latency, however, it may impact eMBB transmission reliability. Therefore, a coexistence mechanism is required to reduce the performance degradation of the ongoing eMBB transmissions.
		\item \textbf{Orthogonal scheduling:} A number of frequency channels are reserved in advance for URLLC traffic in this approach. There are two reservation mechanisms: semi-static reservation and dynamic reservation. In the semi-static scheme, the gNB intermittently broadcasts the frame structure configurations such as frequency numerology. However, in the dynamic reservation, the frame structure information is updated frequently using the control channel of a scheduled user. The downside of this approach is that resources reserved for URLLC will be wasted in case of there is no URLLC transmission. Furthermore, the dynamic scheme needs additional control overhead compared to the semi-static scheme.
	\end{itemize}
	
In this work, we consider the preemptive (puncturing) scheduling approach to handle the dynamic multiplexing\footnote{URLLC traffic is overlapped on eMBB traffic at every mini-slot, which is referred to \emph{dynamic multiplexing} in 3GPP \cite{Park_Wireless_Communications,3gpp_release15}.} of eMBB/URLLC traffic. Under this scenario, the immediate scheduling of URLLC traffic by halting the ongoing eMBB transmission will impact both the throughput and reliability of eMBB transmissions. Therefore, it is imperative to study the coexistence problem of eMBB and URLLC services in 5G NR, where an optimization-based resource allocation problem should not just aim at maximizing the average data rate of eMBB users, but also should consider both eMBB and URLLC reliability. 
% 	Motivated by the aforementioned facts, this work studies the coexistence problem of eMBB and URLLC services in 5G NR considering the puncturing scheduling approach. Specifically, we formulate the coexistence problem as an optimization-based resource allocation problem that aims at maximizing the average data rate of eMBB users while considering both eMBB and URLLC reliability. 
% \vspace{-0.6cm}
	\subsection{Challenges and Contributions}
% 		\vspace{-0.5cm}
	The coexistence of URLLC and eMBB services
on the same radio resource leads to a challenging scheduling problem which is not straightforward to tackle due to the underlying trade-off among latency, reliability and spectral efficiency. In this work, we overcome this problem considering the puncturing scheduling approach. In particular, we develop an efficient mechanism to ensure eMBB data rate and reliability while guaranteeing URLLC stringent requirements. In doing so, we not only focus on the eMBB throughput maximization problem with the URLLC constraints, like the existing works, but also incorporate the risk\footnote{In this work, risk is defined as a measure to quantify the impact of puncturing process on eMBB users data rate when serving the random URLLC traffic.} associated with eMBB transmissions while serving random URLLC requests during the ongoing transmission period in the optimization problem. Specifically, we exploit the variance of eMBB data rate to
	characterize the associated risk and reliability of eMBB transmissions due to the coexistence problem.
	
	Furthermore, most of the existing works in eMBB-URLLC coexistence problem adopt standalone optimization which has its own limitations. First, these techniques are inefficient to well-capture the dynamic characterstics of URLLC traffic and wireless channel conditions. Secondly, due to the complexity of formulated optimization problem, a naive relaxation approach to get the optimal solution may violate the URLLC reliability constraints at the worst case scenario, making these techniques detrimental to network performance.
	
	In  practice,  URLLC  traffic  is  random  and  sporadic;  thus,  it  is  necessary  to  dynamically  and intelligently allocate resources to the URLLC traffic by interacting with the environment. In this regard, DRL approaches can solve non-deterministic problems and make decisions in real-time; thus, making it an appropriate choice to work along with the optimization problem in addressing resource allocation issues and decision making under uncertainty. However, applying DRL in URLLC is not straightforward because the general constraints of quality-of-service (QoS) requirements, such as in the optimization problems, are not present in DRL. Therefore, we have carefully designed reward function based on the traffic scheduling policies that take into account the requirements of eMBB and URLLC services. 
	
	Similarly, DRL may suffer from slow convergence, making it limited to be adopted as a sole solution for handling the aforementioned issues. Thus, in this work, we propose a novel holistic approach combining optimization theory based methods with DRL technique to improve the performance of resource allocation in a dynamic multiplexing scenario of eMBB-URLLC traffic. Specifically, our main contributions are: 
	\begin{itemize}
		\item \textbf{We formulate the resource slicing problem as an optimization problem that maximizes the average data rate of eMBB users, minimizes the variance of eMBB users' data rate, while satisfying the URLLC constraints.} Here, minimizing the variance of eMBB rate reduces the risk on eMBB transmissions, thereby enhancing  its reliability. Furthermore, to ensure a high URLLC reliability, the corresponding reliability constraint is cast as a chance constraint which effectively captures the risk-tail distribution of the outage probability. 
		
		\item \textbf{We propose a two-phase-framework, including eMBB resource allocation and URLLC scheduling phases, that copes with  the dynamic URLLC traffic and channel variations.} In particular, RBs and transmission power are allocated to eMBB users at the eMBB reource allocation phase. Due to the dynamic nature of both URLLC traffic, and channel variations and in order to ensure the reliability requirement of URLLC service, we propose a DRL-based algorithm to schedule the URLLC transmissions over the ongoing eMBB transmissions in the URLLC scheduling phase.   
		
		\item \textbf{In the eMBB resource allocation phase, we first reformulate the optimization problem using the exponential utility function capturing both mean and variance of the eMBB data rate. Then, a Decomposition and Relaxation based Resource Allocation (DRRA) algorithm is proposed.} The proposed DRRA algorithm decomposes the optimization problem into three subproblems: 1) eMBB RBs allocation, 2) eMBB power allocation, and 3) URLLC scheduling. Then each problem is solved individually based on its structure in order to achieve a practical solution with low computation complexity. Specifically, the RBs allocation and power allocation problems are relaxed into convex optimization problems. However, the URLLC resource allocation problem is combinatorial in nature for which it is difficult to achieve a closed-form solution. Hence, we replace the integer variable in the URLLC scheduling problem, i.e., the number of punctured short TTIs (mini-slots), by a continuous weighting variable for each RB. Later, we calculate the number of punctured mini-slots from each RB by modeling it as a binomial distribution with parameters \emph{puncturing weight} and \emph{number of mini-slots in each time slot}.
		
		\item \textbf{In the URLLC scheduling phase, a DRL based algorithm is proposed to cope with URLLC reliability violations, caused due to the relaxation techniques applied in the eMBB resource allocation phase, and to smartly distribute the URLLC traffic on the eMBB users by tackling the dynamics of URLLC traffic and channel variations.} To handle the slow convergence issue of the DRL, we propose a policy gradient based actor-critic learning (PGACL) algorithm that can learn policies by combining the policy learning and value learning with a good convergence rate. Moreover, at the initial start, we leverage the URLLC scheduling results obtained by the DRRA algorithm in the eMBB resource allocation phase to train the PGACL algorithm and improve its convergence time. Hence, combining the advantages of the DRRA and PGACL algorithms (DRRA-PGACL) provides a reliable and efficient resource allocation approach. 
		\item \textbf{The computation complexity of the proposed algorithm is studied in terms of convergence time and accuracy.} Furthermore, extensive simulations are performed to validate our proposed algorithms. Simulation results show that the proposed algorithms can satisfy the stringent URLLC reliability while keeping the eMBB reliability higher than $90\%$. 
	\end{itemize}   
% 	\vspace{-0.7cm}
	\subsection{Organization}
% 		\vspace{-0.5cm}
	Section II summarizes the related works in the following subsections: A. URLL requirements and design, B. Coexistence of eMBB URLLC services, and C. DRL in wireless networks. We present the system model and problem formulation in Section III. Specifically, we introduce the impact on the data rate of eMBB users, the URLLC data rate, chance constraints of URLLC requirements, and the final problem formulation. In Section IV, we present the proposed eMBB resource allocation algorithm. A DRL based resource slicing framework is presented in Section V. We evaluate the performance of the proposed algorithms in Section VI. Section VII concludes the paper.  
% 	\vspace{-0.7cm}
	\section{Related works}
% 		\vspace{-0.4cm}
	\subsection{URLLC requirements and design}
% 		\vspace{-0.5cm}
	Research works focusing on URLLC are gaining attention in both academia and industry. For example, the work in \cite{Park_Wireless_Communications} highlighted the key requirements of URLLC and its physical layer issues. The authors presented enabling technologies for URLLC  in 5G NR such as packet structure, frame structure, and scheduling schemes discussed in 3GPP Release 15 standardization. In \cite{Popovski2019}, the authors discussed communication-theoretic principles for the design of URLLC including the medium access control (MAC) protocols, massive MIMO, interface-diversity, and multi-connectivity. The authors of \cite{Park2020} discussed the limitations of 5G URLLC and provided key research directions for the next generation of URLLC, named eXtreme URLLC (xURLLC). The authors proposed three concepts for the xURLLC: 1) Predicting channels, traffic, and other key performance indicators by leveraging the machine learning technology; 2) Fusing both radio frequency and non-radio frequency modalities for predicting rare events; and 3) Joint communication and control co-design. The study conducted in \cite{Sun2019} discussed the resource allocation for URLLC problem considering the achievable rate in the short block-length regime. The resource allocation problem is to optimize the bandwidth allocation, power control, and antenna configuring considering both latency and reliability constraints. The work in \cite{Liu2018} studied the power minimization subject to latency and reliability constraints considering a Manhattan mobility model in Vehicle-to-Vehicle (V2V) networks. The reliability measure is defined in terms of maximal queue length among all vehicle pairs and the extreme value theory is applied to locally characterize the maximal queue length. In \cite{Mei2018}, the authors studied the joint optimization of radio resources, power control, and modulation schemes of the V2V communications while guarantying the latency and reliability requirements of vehicular users and maximizing the rate of cellular users. They used Lagrange dual decomposition and binary search methods to find the optimal solution of the joint optimization problem. 
% 	\vspace{-0.5cm}
	\subsection{Coexistence of eMBB URLLC services} 
% 	\vspace{-0.5cm}
	The authors in \cite{JSAC_2019} explored eMBB and URLLC services in cloud radio access networks. A multi-cast transmission is considered for eMBB slices while URLLC slices are relied on uni-cast transmission. They proposed a generic revenue framework for radio access network slicing and formulated the revenue maximization problem as a mixed-integer nonlinear programming. Semi-definite relaxation is leveraged to solve the optimization problem. In \cite{Anand_Joint}, the authors studied the impact of URLLC traffic on eMBB transmissions modeling the loss of eMBB data rate  associated with URLLC traffic as a linear, convex, or threshold model. The work in \cite{Jihong_VR}  studied the problem of concurrent support of visual and haptic perceptions over wireless cellular networks. The visual traffic is linked to eMBB slices while  the haptic traffic is linked to URLLC slices leading to eMBB-URLLC  multi-modal transmissions. The authors in \cite{Bairagi2020} proposed the PSUM algorithm and transportation model to solve resource scheduling problem for eMBB and URLLC users over time slots and mini-slots, respectively. 
	
	Moreover, the study in \cite{Kassab_Globecom} discussed the performance trade-offs between orthogonal and non-orthogonal multiple access for multiplexing of eMBB and URLLC users in the uplink of a multi-cell cloud radio access networks architecture. The analysis includes orthogonal and non-orthogonal multiple access with different decoding architectures, such as successive interference cancellation and puncturing. \emph{The results show that the orthogonal multiple access approach reduces the eMBB-URLLC mutual interference;  however, URLLC users suffer from the errors caused by packet drops due to the insufficient number of transmission opportunities}. Moreover, the results show significant gains accrued by the successive interference cancellation scheme of URLLC traffic at the edge for non-orthogonal multiple access. Furthermore, the work shows the potential benefits of puncturing in improving the efficiency of fronthaul usage by discarding received mini-slots (short TTIs) affected by URLLC interference. The work in \cite{Popovski_Access} proposed a communication-theoretic model for eMBB, mMTC, and URLLC services considering traffic dynamics that are inherent to each individual service. The authors analyzed the performance of both orthogonal and non-orthogonal slicing. The study demonstrated that the non-orthogonal slicing scheme can ensure performance level for all services by leveraging their heterogeneous requirements. The authors in \cite{Bairagi2019} used a heuristic algorithm and one-sided matching game to solve the coexistence problem of eMBB and URLLC traffics. In \cite{Pandey2019}, the authors tried to maximize the data rate of eMBB users while maintaining the reliability requirement of URLLC via solving a multi-armed bandit problem. In our previous work \cite{Alsenwi2019}, we proposed a risk-sensitive formulation based on the Conditional Value at Risk (CVaR) as a risk measure for eMBB reliability and a chance constraint to encode the reliability constraint of URLLC. 
% 	\vspace{-0.7cm}
	\subsection{DRL in wireless networks} 
% 	\vspace{-0.5cm}
	Recently, many works have used the DRL to solve the resource allocation problem and decision making in wireless networks \cite{Fu2018}. The study in \cite{Yang2019} proposed an actor-critic RL model for joint communication mode selection, Resource Block (RB) allocation, and power allocation in device-to-device-enabled V2V based internet of vehicle communication networks. Their objective was to satisfy URLLC requirements of V2V links while maximizing the rate of vehicle-to-infrastructure links. In \cite{Yang2019a}, the authors presented a heterogeneous radio frequency/visible light communication industrial network architecture. They formulated a joint uplink and downlink resource management decision-making problem as a Markov decision process. The work in \cite{Kasgari2019} presented a deep RL model to provide URLLC in the downlink of an orthogonal frequency division multiple access network. The problem is formulated as a power minimization problem with rate, latency, and reliability constraints. The rate of each user is calculated and mapped to the RB and power allocation vectors in order to solve the problem using deep RL algorithm. The latency and reliability of each user are  used as a feedback to the deep RL algorithm. In \cite{DRL_eMBB_URLLC_1} and \cite{DRL_eMBB_URLLC_2}, the authors proposed a DRL based algorithms to solve the coexistence problem of eMBB and URLLC. A deep deterministic policy gradient based method is used in  \cite{DRL_eMBB_URLLC_1} while the deep Q-learning algorithm is leveraged in  \cite{DRL_eMBB_URLLC_2}.
	
	Unlike related works, we not only focus on the eMBB throughput maximization problem with the URLLC constraints, but also incorporate the risk associated with eMBB transmissions while serving random URLLC request during the ongoing transmission period in the optimization problem. In particular, we exploit the variance of eMBB data rate to characterize the associated risk and reliability of eMBB transmissions due to the coexistence problem. Furthermore, we propose a holistic approach combining optimization theory based methods with DRL technique to improve the performance of resource allocation in a dynamic multiplexing scenario of eMBB-URLLC traffic.  Therefore, this work is, to the best of our knowledge, the first to analyze the mean-variance aspects of eMBB-URLLC coexistence problem by combining optimization theory based methods with the DRL. 
% 		\vspace{-0.7cm}
	\section{System Model and Problem Formulation}
		\begin{figure}[t]
		\centering
		\includegraphics[width=0.5\linewidth]{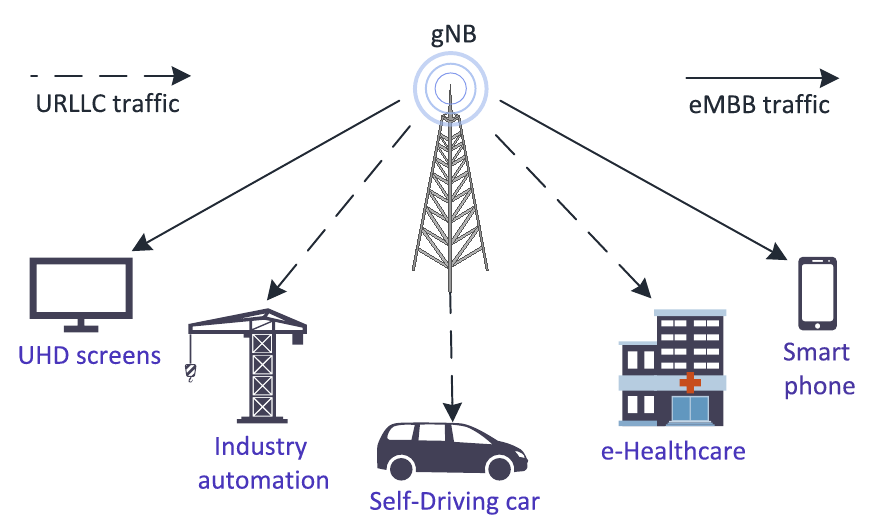}
		\caption{System model.}
		\label{system_model}
	\end{figure}
	\begin{table}[h!]
		\caption{Summary of Notations}
		\centering
		\fboxrule=0pt
		\begin{tabular}{ll}
			\toprule 
			Notation & Definition \\
			\midrule
			$\mathcal{K}, \mathcal{N}, \mathcal{B}$ & Set of eMBB users, URLLC users, and RBs, respectively \\ 
			$x_{kb}(t), p_{kb}(t), z_{kb}(t)$ & RBs allocation, Power allocation, and Puncturing variables, respectively for $k\in\mathcal{K}$, and $b\in\mathcal{B}$\\
			$r_k^e(t)$ & Data rate of eMBB user $k$ at time slot $t$\\ 
			$r_{n}^u(t)$ & Data rate of URLLC user $n$ at time slot $t$\\
			$h_{kb}^e(t)$ & eMBB channel gain,  for $k\in\mathcal{K}$, and $b\in\mathcal{B}$\\
			$h_{nb}^u(t)$ & URLLC channel gain,  for $n\in\mathcal{N}$, and $b\in\mathcal{B}$\\
			$p_{nb}^u(t)$ & URLLC transmission power, for $n\in\mathcal{N}$, and $b\in\mathcal{B}$ \\
			$L(t)$ & Total number of URLLC packets at a time slot $t$\\
			$c_{nb}^u(t)$ & Length of CB in symbol, for $n\in\mathcal{N}$, and $b\in\mathcal{B}$\\ 
			$D_{nb}^u(t)$ & Channel dispersion at time slot $t$ \\
			$P_{\textrm{max}}$ & Maximum transmission power \\
			$M$ & Number of mini-slots in an eMBB time slot \\
			$\mu$ & Parameter controls the desired-risk sensitivity of $\tilde{g}_k$\\
			$\Theta^{\mathsf{max}}$ & Maximum allowed outage probability of URLLC traffic\\ 
			$\zeta$ & URLLC packet size\\
			$f_b$& Bandwidth of RB $b$\\
			$\sigma^2$ & Noise power\\
			$\alpha, \beta$ & Weighting parameters \\
			$\mathcal{A}, \mathcal{S}$ & Set of action space and state space, respectively\\
			$R(\boldsymbol{a, s})$ & Reward function, for $a(t)\in\mathcal{A}$, and  $s(t)\in\mathcal{S}$ \\
			$\phi(t)$& Time-varying weights for URLLC reliability\\
			$\pi$& Puncturing policy\\
			$Q^{\pi}(\boldsymbol{a,s})$ & Cumulative discounted reward at a given $\pi$\\
			$J(\pi)$ & Network objective reward value\\
			$V(\boldsymbol{a, s})$ &Value function of the agent $k$ \\
			$\rho_a, \rho_c$ & Actor and critic learning rate, respectively\\
			\bottomrule  
		\end{tabular}
	\end{table} 
	We consider two types of downlink requests, i.e., URLLC slice and eMBB slice requests. As shown in Fig. \ref{system_model}, there are different types of users connected to a gNB such as self-driving cars, smartphones, industrial automation, etc. We consider a gNB serving a set $\mathcal{K}$ of $K$ eMBB users and a set $\mathcal{N}$ of $N$ URLLC users. Let $B$ denote the total number of RBs, where a RB $b\in\mathcal{B}=\{1, 2, \ldots, B\}$ occupies $12$ sub-carrier in frequency. The summary of notations used in this work is presented in Table 1.
	
	Typically, eMBB transmissions are allowed to span multiple time resources in order to increase spectrum efficiency. However, URLLC transmissions are localized in time domain and can span multiple frequency channels due to its latency requirements. Moreover, the arriving URLLC traffic during the eMBB transmission cannot be delayed until completing eMBB transmissions due to its hard latency constraints. Thus, we schedule URLLC traffic and transmit it immediately by puncturing the ongoing eMBB transmissions. In reality, puncturing (preemption) is done by the gNB scheduler\footnote{For multiplexing between eMBB and URLLC traffics, 3GPP release 15 proposes \emph{the preemption indication (PI)} \cite{3gpp_release15}.}. In this work, we consider that URLLC users are scheduled with short TTI (mini-slot), while eMBB users are scheduled with long TTI size (slot of $1$ ms duration) \cite{Park_Wireless_Communications}. Fig.\ref{multiplexing_eMBB_URLLC} shows the ongoing eMBB transmission with a long TTI duration, where the incoming URLLC packet preempts a part of the eMBB transmissions. As shown in Fig. \ref{multiplexing_eMBB_URLLC}, the transport block of eMBB user 2 consists of seven code-blocks, where each code-block is mapped sequentially to the scheduled time-frequency resources. When the URLLC service is initiated in the second and sixth cod-blocks of the transport block of eMBB user 2, the symbols in these code-blocks are replaced by the symbols of URLLC packets which degrades the quality of the eMBB service. This problem is a serious concern to eMBB service, thus a proper mechanism to protect the ongoing eMBB transmissions should be introduced.
		\begin{figure*}[t!]
		\centering
		\includegraphics[width=0.75\linewidth]{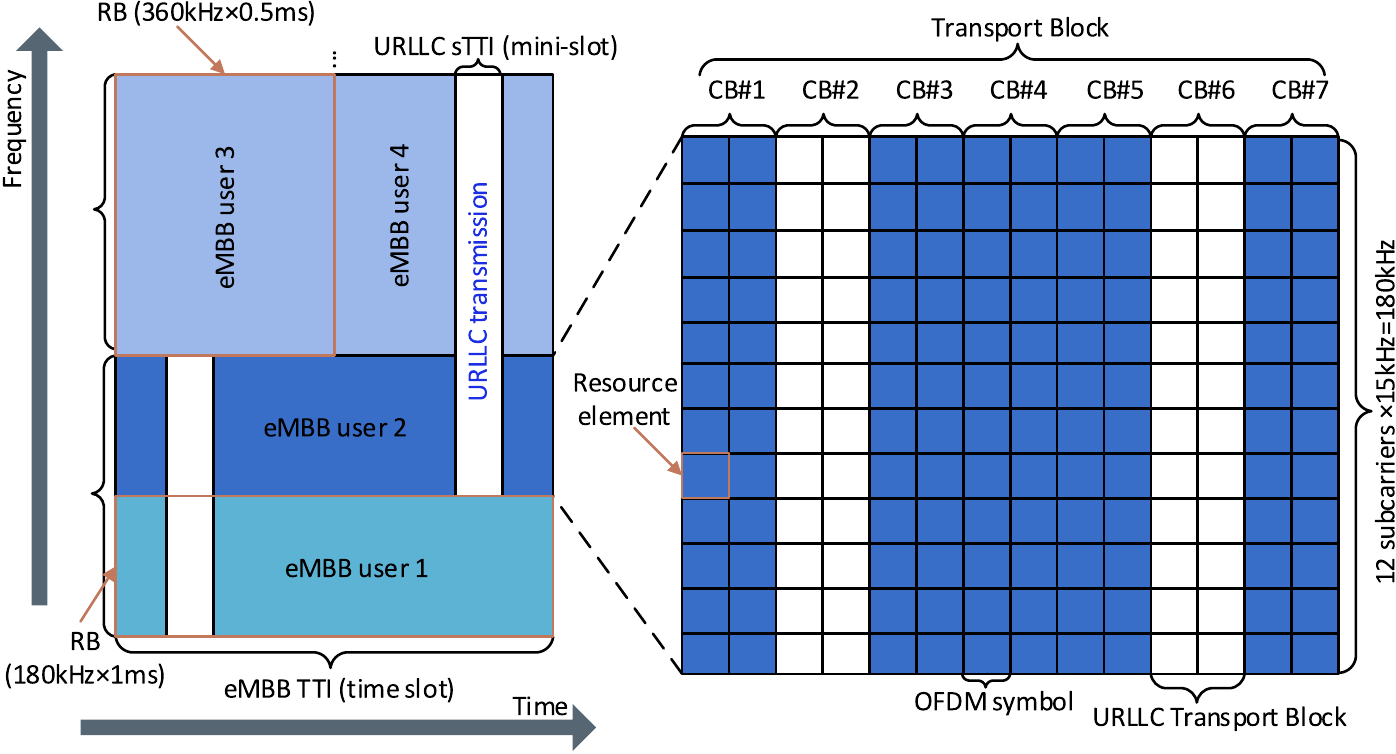}
		\caption{Multiplexing of eMBB/URLLC traffics.}
		\label{multiplexing_eMBB_URLLC}
% 			\vspace{-0.4cm}
	\end{figure*}
% \setlength{\textfloatsep}{6pt}
% \vspace{-0.7cm}
	\subsection{eMBB data rate based on Shannon capacity model}
% 		\vspace{-0.5cm}
	Puncturing eMBB transmissions by URLLC traffic impacts the data rate of eMBB users. Let $z_{kb}(t)$ be the number of punctured mini-slots from the RB $b$ of eMBB user $k$ at time slot $t$. Accordingly, the data rate of an eMBB user $k$ over a RB $b$ at time slot $t$ can be approximated as 
	\begin{equation}
	r_{kb}^e(t)=f_b\Big(1-\frac{z_{kb}(t)}{M}\Big)\log_2\left(1+\frac{p_{kb}(t)h_{kb}(t)}{\sigma^2}\right),
	\end{equation}
	where $f_b$ is the bandwidth of the RB $b$, $M$ is the number of mini-slots in an eMBB time slot, $h_{kb}(t)$ is the time-varying Rayleigh fading channel gain of the transmission, and $p_{kb}(t)$ is the downlink transmission power of the gNB on the RB $b$ to user $k$ at slot $t$. Therefore, the data rate of the eMBB user $k$ over all allocated RBs can be given as
	\begin{equation}
	r_{k}^e(t)=\sum\limits_{b\in\mathcal{B}}x_{kb}(t)r_{kb}^e(t),
	\end{equation}
	where $x_{kb}(t)$ is the eMBB user scheduling indicator at time slot $t$ defined as follows:
	\begin{equation}
	x_{kb}(t)=\begin{cases} 1, &\mbox{if the RB} \;b \;\mbox{is allocated to user}\; k \mbox{ at time }t,\\
	0, & \mbox{otherwise. }
	\end{cases}
	\label{RB_Allocaiton}
	\end{equation}
% 	\vspace{-0.7cm}
	\subsection{URLLC data rate based on finite block-length coding}
% 		\vspace{-0.5cm}
	In URLLC, packets are typically very short, and thus, the achievable rate and the transmission error probability cannot be accurately captured by Shannon's capacity. Instead, the achievable rate in URLLC falls in the finite block-length channel coding regime, which is derived in \cite{finite_block_lenght}. Let $r_{n}^{u}(t)$ be the achievable rate of URLLC user $n$ at time slot $t$ and $c_{nb}^u(t)$ be the length of the CB in symbols (i.e., number of symbols in a mini-slot). We consider that the Frequency Division Duplex (FDD) is applied inside the URLLC resources. Thus, the URLLC data rate can be given by \cite{finite_block_lenght}:
	\begin{equation}
	r_{n}^u(t)=\sum\limits_{k\in\mathcal{K}}\sum\limits_{b\in\mathcal{B}}\frac{f_bx_k^b(t)z_{kb}(t)}{M\times N}\log\left(1+\frac{p_{nb}^u(t)h_{nb}^u(t)}{\sigma^2}\right)
	-\sqrt\frac{D_{nb}^u(t)}{c_{nb}^u(t)}Q^{-1}(\vartheta),
	\label{URLLC_data_rate}
	\end{equation}
	where $Q^{-1}(\cdot)$ is the inverse of the Gaussian Q-function, $\vartheta>0$ is the transmission error probability, and $D_{nb}^u(t)$ represents the characteristic of the channel called the \textit{channel dispersion}, i.e., $D_{nb}^u(t)$ determines the stochastic variability of the channel of user $n$ at time sot $t$ relative to a deterministic channel with the same capacity, given by 
	\begin{equation}
	D_{nb}^u(t)=1-\frac{1}{\left(1+\frac{p_{n}^u(t)h_{n}^u(t)}{\sigma^2}\right)^2}. 
	\end{equation}
% 	\vspace{-0.7cm}
	\subsection{Problem formulation}
% 		\vspace{-0.5cm}
	% \subsection{eMBB Resource Allocation}
	We allocate RBs and transmission power to eMBB users at the beginning of each eMBB time slot. Then, we schedule the incoming URLLC traffic on the ongoing eMBB transmissions by puncturing some resources from eMBB users. Generally, puncturing eMBB users with low data rate (users located at bad channel conditions like at the cell edge) causes high degradation on eMBB reliability\footnote{In this work, eMBB reliability is defined based on the minimum data rate satisfaction of eMBB users, i.e., the percentage of eMBB users that get the minimum required data rate out of the total number of eMBB users.} which should be considered when designing a reliable resource allocation framework. Thus, the proposed resource allocation strategy aims at: \emph{1) maximizing the average eMBB data rate, 2) reducing the impact on eMBB reliability, and 3) satisfying the URLLC constraints}. Due to the uncertainty in wireless channels, we propose a risk-averse formulation by considering the variance of eMBB data rate, in addition to the average eMBB data rate, so as to satisfy the minimum data rate of each eMBB user and enhance its reliability. In this regard, moving from the conventional average based formulation to the risk-averse formulation will reduce the impact on the eMBB reliability that comes from the variations in the wireless channel quality and URLLC scheduling. Analogous to risk-averse formulations in modern portfolio theory (MPT) \cite{markowitz2000mean}, in a dynamic multiplexing scenario of eMBB-URLLC traffics, the gNB needs to construct an appropriate puncturing preferences (\textit{similar to investment strategy}) on eMBB users when serving the incoming URLLC traffic. Therefore, we define a function that captures both the average sum of eMBB data rate and its variance as 
	\begin{equation}
	\mathcal{F}(\boldsymbol{x}, \boldsymbol{p}, \boldsymbol{z})=\sum\limits_{k=1}^{K}\mathbb{E}_{\text{h}}\bigg[\frac{1}{T}\sum\limits_{t=0}^{T}r_{k}^e(t)\bigg]-\beta \textsf{Var}_{\text{h}}\big[r_k^e(t)\big],  
	\label{URLLC_objjective_function}
	\end{equation}
	where $\mathbb{E}$ refers to the expectation, $\textsf{Var}$ refers to the variance, and $\beta$ is the variance weight. The variance part captures the dynamic characteristics of wireless channels to define the reliability of eMBB, as it efficiently characterizes the risk of investments in MPT. 
% 	\textcolor{blue}{The variance part in the risk-averse formulation well-captures the dynamic characteristics of wireless channels during the puncturing process. Thus, it can appropriately reflect the reliability of eMBB traffic.}

	On the other hand, the URLLC reliability can be achieved by ensuring that its outage probability is less than a threshold $\Theta^{\mathsf{max}}$, where $\Theta^{\mathsf{max}}$ is a small positive value ($\Theta^{\mathsf{max}}<<1$). Let $L_m(t)$ be a random variable denoting the number of arrived URLLC packets at a minislot $m\in\mathcal{M}=\{1, 2, ..., M\}$ of the time slot $t$, and $L(t)=\sum_{m\in\mathcal{M}}L_m(t)$ is the total number of arrived URLLC packets in the time slot $t$. Then, the URLLC reliability constraint can be defined as
	\begin{equation}
	\textsf{Pr}\left[\sum\limits_{n\in\mathcal{N}}r_n^u(t)\leq \zeta L(t)\right]\leq\Theta^{\mathsf{max}},
	\end{equation}
	where $\zeta$ is the URLLC packet size. 
	
	Accordingly, the joint eMBB/URLLC resource allocation problem can be formulated as follows:
	\begin{maxi!}[2]                 % mini! = minimize 
		{\boldsymbol{x, p, z}}                               % optimization variable
		{\sum\limits_{k=1}^{K}\mathbb{E}_{\text{h}}\bigg[\frac{1}{T}\sum\limits_{t=0}^{T}r_{k}^e(t)\bigg]-\beta \textsf{Var}_{\text{h}}\big[r_k^e(t)\big]}{\label{main_problem}}{}   
		\addConstraint{\textsf{Pr}\bigg[\sum\limits_{n=1}^{N}r_n^u(t)\leq \zeta L(t)\bigg]}{\leq\Theta^{\mathsf{max}}, \label{URLLC_rel_const_main}}
		%	\addConstraint{Pr\bigg[d_m\leq\frac{L_m\times q_n}{r_n^u(t)}}{\bigg]\leq\epsilon_t^u, \; \forall n\in\mathcal{N}, \label{URLLC_time_Const_main}} 
		\addConstraint{\sum\limits_{k=1}^{K}\sum\limits_{b=1}^{B}p_{kb}(t)}{\leq P_{\textrm{max}}, \label{main_const_1}}
		\addConstraint{\sum\limits_{k=1}^{K}x_{kb}(t)}{\leq 1,\;\; \forall \; b\in \mathcal{B}, \label{main_const_2}}
		\addConstraint{p_{kb}(t)}{\geq 0,\; \forall k\in\mathcal{K},\; b\in\mathcal{B}, \label{main_const_3}}
		\addConstraint{x_{kb}(t)}{\in \{0, 1\}, \;\; \forall k\in \mathcal{K},\;\; b \in \mathcal{B},\label{main_const_4}}
		\addConstraint{z_{kb}(t)}{\in\{0, 1, \dots, M\}, \;\; \forall k \in \mathcal{K},\; b\in\mathcal{B}, \label{main_const_5} }
	\end{maxi!}
	where $P_{\textrm{max}}$ is the maximum transmission power of the gNB. The optimization problem (\ref{main_problem}) seeks the optimum RBs allocation matrix to eMBB users $\boldsymbol{x^*}$, the optimum power allocation vector to eMBB users $\boldsymbol{p^{*}}$, and the optimum number of punctured mini-slots of all RBs matrix $\boldsymbol{z^*}$. The objective function is formulated based on Markowitz mean-variance model \cite{markowitz2000mean} to maximize the average eMBB data rate for a given level of risk. The probability constraint \eqref{URLLC_rel_const_main} ensures the URLLC reliability. Furthermore, constraints \eqref{main_const_1}, \eqref{main_const_2}, \eqref{main_const_3}, and \eqref{main_const_4} represent the RBs and power allocation constraints. Finally, constraint \eqref{main_const_5} ensures that the number of punctured mini-slots form a RB $b$ can take any integer number less than $M$. In this paper, we consider that the gNB transmits with maximum allowed power to URLLC users in order to enhance the URLLC transmission reliability.  
	
	The optimization problem (\ref{main_problem}) is a mixed-integer nonlinear programming (MINLP) and NP-hard problem. To find a global optimum solution, we need to search the space of feasible URLLC placement mini-slots with all possible combinations of eMBB user RBs allocation and power allocation. This may require exponential-complexity to solve. To avoid this difficulty, we propose a two-phase approach based on optimization methods and learning in the next two sections.  
% 		\vspace{-0.4cm}
	\section{\MakeLowercase{e}MBB Resource Allocation: Optimization Methods Based Approach}
% 		\vspace{-0.4cm}
	In this section, we first simplify the objective function in \eqref{main_problem} to a smoothing form and eliminate the complexity caused by the variance, i.e., the variance involves the term $\left(\mathbb{E}_h[r_k^e(t)]\right)^2$, by using an equivalent risk-averse utility function. We consider the exponential function that can capture both the mean and variance as defined in \cite{Risk_Sensitive_2002}: 
	\begin{equation}
	\mathcal{G}(\boldsymbol{x,p,z})=\frac{1}{\mu}\log\mathbb{E}_\textrm{h}\left[\exp\biggl(\mu\sum\limits_{k=1}^{K}r_{k}^e(t)\biggr)\right],
	\label{utility_function}
	\end{equation}
	where the parameter $\mu$ controls the desired risk-sensitivity. The utility function \eqref{utility_function} becomes a strongly concave when increasing the values of $\mu$ negatively reflecting more risk-averse tendencies. Furthermore, the utility function \eqref{utility_function} becomes a risk-neutral at  $\mu\to 0$. The Taylor expansion of the exponential utility function around $\mu=0$ is given as
	\begin{equation}
	\mathcal{G}(\boldsymbol{x,p,z})=\mathbb{E}_\textrm{h}\left[\sum\limits_{k=1}^{K}r_k^e(t)\right]+\frac{\mu}{2}\textsf{Var}\left[\sum\limits_{k=1}^{K}r_{k}^e(t)\right]+\mathcal{O}(\mu^2).
	\label{Taylor_expansion}
	\end{equation}
	
	Equation \eqref{Taylor_expansion} shows that the utility function in \eqref{utility_function} effectively captures both mean and variance terms of eMBB users' data rate. Accordingly, we can obtain an equivalent formulation of \eqref{main_problem} as follows:
	\begin{maxi!}[2]                 % mini! = minimize 
		{\boldsymbol{x, p, z}}                               % optimization variable
		{\frac{1}{\mu}\log\mathbb{E}_\textrm{h}\left[\exp\biggl(\mu\sum\limits_{k=1}^{K}r_{k}^e(t)\biggr)\right]}{\label{modified_problem}}{\textbf{P: }}   
		\addConstraint{\textsf{Pr}\bigg[\sum\limits_{n=1}^{N}r_n^u(t)\leq \zeta L(t)\bigg]}{\leq\Theta^{\mathsf{max}}, \label{URLLC_rel_const_main2}}
		%	\addConstraint{Pr\bigg[d_m\leq\frac{L_m\times q_n}{r_n^u(t)}}{\bigg]\leq\epsilon_t^u, \; \forall n\in\mathcal{N}, \label{URLLC_time_Const_main}} 
		\addConstraint{\sum\limits_{k=1}^{K}\sum\limits_{b=1}^{B}p_{kb}(t)}{\leq P_{\textrm{max}}, \label{modified_const_1}}
		\addConstraint{\sum\limits_{k=1}^{K}x_{kb}(t)}{\leq 1,\;\; \forall \; b\in \mathcal{B}, \label{modified_const_2}}
		\addConstraint{p_{kb}(t)}{\geq 0,\; \forall k\in\mathcal{K},\; b\in\mathcal{B}, \label{modified_const_3}}
		\addConstraint{x_{kb}(t)}{\in \{0, 1\}, \;\; \forall k\in \mathcal{K},\; \; b \in \mathcal{B},\label{modified_const_4}}
		\addConstraint{z_{kb}(t)}{\in\{0, 1, \dots, M\}, \;\; \forall k \in \mathcal{K},\; b\in\mathcal{B}. \label{modified_const_5} }
	\end{maxi!}
	
	Note that \textbf{P} is still a mixed-integer problem which is a non-convex problem. To solve \textbf{P}, we propose a decomposition and relaxation based resource allocation (DRRA) algorithm. In this algorithm, we first decompose \textbf{P} into three sub-problems: \textbf{P1:} eMBB RBs allocation, \textbf{P2:} eMBB power allocation, and \textbf{P3:} URLLC scheduling. Then, we relax $\boldsymbol{x}$ and $\boldsymbol{z}$ to continuous variables in $\tilde{\textbf{P}}\textbf{1}$ and $\tilde{\textbf{P}}\textbf{3}$, respectively. Moreover, the probability constraint \eqref{URLLC_rel_const_main2} is relaxed to a linear constraint using the Markov's inequality. Next, we iteratively solve $\tilde{\textbf{P}}\textbf{1}$, \textbf{P2}, and $\tilde{\textbf{P}}\textbf{3}$ till convergence. Finally, we perform a binary conversion techniques to meet constraint \eqref{modified_const_4} as shown in Algorithm \ref{Algorithm}. 
% 	\vspace{-0.7cm}
	\subsection{eMBB RBs allocation problem}
% 		\vspace{-0.5cm}
	For any fixed feasible URLLC placement $\boldsymbol{z}$ and $\boldsymbol{p}$, the problem \textbf{P} can be represented as follows:
% 		\vspace{-0.7cm}
	\begin{maxi!}[2]                 % mini! = minimize 
		{\boldsymbol{x}}                               % optimization variable
		{\frac{1}{\mu}\log\mathbb{E}_\textrm{h}\left[\exp\biggl(\mu\sum\limits_{k=1}^{K}r_{k}^e(t)\biggr)\right]} {\label{RB_Allocation_main}}{\textbf{P1: }}
		\addConstraint{\sum\limits_{k=1}^{K}x_{kb}(t)}{\leq 1,\;\; \forall \; b\in \mathcal{B}, \label{const2}} 
		\addConstraint{x_{kb}}{\in \{0, 1\}, \;\; \forall k\in \mathcal{K}\; \text{and}\; b \in \mathcal{B}.}
	\end{maxi!}
	
	The optimization problem (\ref{RB_Allocation_main}) is an integer nonlinear programming (MINLP) which can be relaxed to a  problem whose solution is within a constant approximation from the optimal. The fractional solution is then rounded to get a solution to the original integer problem. Accordingly, the optimization problem (9) can be approximated as follows:
	\begin{maxi!}[2]                 % mini! = minimize 
		{\tilde{\boldsymbol{x}}}                               % optimization variable
		{\frac{1}{\mu}\log\mathbb{E}_\textrm{h}\left[\exp\biggl(\mu\sum\limits_{k=1}^{K}r_{k}^e(t)\biggr)\right]\label{eMBB_RB_Allocaiton_objective2} }  % objective function and label
		{\label{RB_Allocaiton_Relaxed}}{\tilde{\textbf{P}}\textbf{1}: }             % label for optimizatio problem 
		% optimization result
		\addConstraint{\sum\limits_{k=1}^{K}\tilde{x}_{kb}(t)}{\leq 1,\;\; \forall \; b\in \mathcal{B}, \label{const_a_realxed}} 
		\addConstraint{0\leq \tilde{x}_{kb}(t)}{\leq 1, \;\; \forall k\in \mathcal{K},\;\; b \in \mathcal{B}. \label{const_b_realxed}}
		%\addConstraint{0\leq\omega_{u}\leq 1,}{\;\; \forall u \in \mathcal{U}}
	\end{maxi!}
	
	\begin{lemma}
		For a given $\boldsymbol{p}$ and $\boldsymbol{z}$, \eqref{RB_Allocaiton_Relaxed} is a convex optimization problem.
	\end{lemma}
	\begin{proof}
		We prove the convexity of \eqref{RB_Allocaiton_Relaxed} in two steps. First, we prove that the objective function $\mathcal{G}(\cdot)$ is concave with respect to $\tilde{\boldsymbol{x}}$. Then, we prove  the convexity of the feasible region. Here, we can notice that $r_k^e(\tilde{\boldsymbol{x}})$ is a linear function in $\tilde{\boldsymbol{x}}$ for $0\leq \tilde{\boldsymbol{x}}\leq 1$. Moreover, using the scalar composition property  in convexity,  we have logarithmic of a convex function to be a concave. Next, the constraints \eqref{const_a_realxed} and \eqref{const_a_realxed} are linear constraints. Therefore, \eqref{RB_Allocaiton_Relaxed} is a convex optimization problem.
	\end{proof}
	Later, we use the threshold rounding technique described in \cite{rounding_1} to enforce the relaxed $\boldsymbol{x}$ to be a binary variable. Let $\eta\in[0,1]$ be a rounding threshold. Then, we set $x_{kb}^*$ as 
	\begin{equation}
	x_{kb}^*=\begin{cases} 1, &\mbox{if} \;\tilde{x}_{kb}^*\geq\eta,\\
	0, & \mbox{otherwise. }
	\end{cases}
	\label{rounding_threshold}
	\end{equation}
	
	The binary solution obtained from \eqref{rounding_threshold} may violate RB allocation constraint. To overcome this issue, we modify problem \eqref{RB_Allocaiton_Relaxed} as follows:
	\begin{maxi!}[2]                 % mini! = minimize 
		{\tilde{\boldsymbol{x}}}                               % optimization variable
		{\frac{1}{\mu}\log\mathbb{E}_\textrm{h}\left[\exp\biggl(\mu\sum\limits_{k=1}^{K}r_{k}^e(t)\biggr)\right]+\alpha\mathrm{\Delta} \label{eMBB_RB_Allocaiton_objective_Rounding} }  % objective function and label
		{\label{RB_Allocaiton_Rounding}}{}             % label for optimizatio problem 
		% optimization result
		\addConstraint{\sum\limits_{k=1}^{K}\tilde{x}_{kb}(t)}{\leq 1+\mathrm{\Delta},\;\; \forall \; b\in \mathcal{B}, \label{const_a}} 
		\addConstraint{0\leq \tilde{x}_{kb}(t)}{\leq 1, \;\; \forall k\in \mathcal{K},\; \; b \in \mathcal{B},\label{const_b}}
		%\addConstraint{0\leq\omega_{u}\leq 1,}{\;\; \forall u \in \mathcal{U}}
	\end{maxi!}
	where $\mathrm{\Delta}$ is the maximum violation of RB allocation constraint given as
	\begin{equation}
	\mathrm{\Delta}=\max\{0, \sum\limits_{k\in\mathcal{K}}x_{kb}-1\}, \; \forall b\in\mathcal{B},
	\end{equation}
	and $\alpha$ is the weight of $\mathrm{\Delta}$ which take a negative value. The problem (15) aims at maximizing $\mathcal{G}(\tilde{\boldsymbol{x}})$ while minimizing the rounding error $\Delta$. Thus, the feasible solution of \eqref{RB_Allocation_main} is obtained at $\mathrm{\Delta}=0$.
% 	\vspace{-0.7cm}
	\subsection{eMBB power allocation problem}
% 		\vspace{-0.5cm}
	For any given $\tilde{\boldsymbol{x}}$ and $\boldsymbol{z}$, the power allocation problem can be given as 
	\begin{maxi!}[2]                 % mini! = minimize 
		{\boldsymbol{p}}                               % optimization variable
		{\frac{1}{\mu}\log\mathbb{E}_\textrm{h}\left[\exp\biggl(\mu\sum\limits_{k=1}^{K}r_{k}^e(t)\biggr)\right]}{\label{power_allocation_problem}}{\textbf{P2: }} 
		%	\addConstraint{\mathbb{E}\Big[r_n^u(t)\Big]}{\geq\frac{\mathbb{E}(L_m)\times q_n}{d_m\times \epsilon_t^u}, \; \forall n\in\mathcal{N}, m\in\mathcal{M}, \label{URLLC_time_Const2}}
		\addConstraint{\sum\limits_{k=1}^{K}\sum\limits_{b=1}^{B}p_{kb}(t)}{\leq P_{\textrm{max}}, \label{const_1_p}}
		\addConstraint{p_{kb}(t)}{\geq 0,\; \forall k\in\mathcal{K},\; b\in\mathcal{B} \label{const_2_p}}.
	\end{maxi!}
	
	\begin{lemma}
		For a given $\tilde{\boldsymbol{x}}$ and $\boldsymbol{z}$, \eqref{power_allocation_problem} is a convex optimization problem.
	\end{lemma}
	\begin{proof}
		We first prove the convexity of $r^e_k(t)$ with respect to $p_k$ by calculating the second derivative as
		\begin{equation}
		\frac{\partial^2 r_k^e(t)}{\partial p_k^2(t)}=\frac{-x_{kb}f_b(1-z_{kb}/M)(h^e_{kb}/\sigma^2)^2}{\Big(1+\frac{p_{kb}h_{kb}^e}{\sigma^2}\Big)^2},
		\end{equation}
		which is always negative for any value of $p_{kb}$. Thus, combining $r_k^e(t)$ with $\exp$ and $\log$ functions results in a concave function. Moreover, constraints \eqref{const_1_p} and \eqref{const_2_p} are linear constraints. Therefore, \eqref{power_allocation_problem} is a convex optimization problem.
	\end{proof}
% 	\vspace{-0.7cm}
	\subsection{URLLC scheduling problem}
% 		\vspace{-0.5cm}
	%%%%%%%%%%%%%%%%%%%%%%%%%%%%%%%  ALGO  %%%%%%%%%%%%%%%
	For a given $\tilde{\boldsymbol{x}}$ and $\boldsymbol{p}$ the URLLC scheduling problem can be given as
	\begin{maxi!}[2]                 % mini! = minimize 
		{\boldsymbol{z}}                               % optimization variable
		{\frac{1}{\mu}\log\mathbb{E}_\textrm{h}\left[\exp\biggl(\mu\sum\limits_{k=1}^{K}r_{k}^e(t)\biggr)\right]}{\label{URLLC_Resource_Allocaiton_Problem}}{\textbf{P3: }}   
		\addConstraint{\text{Pr}\biggl(\sum\limits_{n=1}^{N}r_n^u(t)\leq \zeta L(t)\biggr)}{\leq\Theta^*, \label{URLLC_reliability_Const3}}
		%	\addConstraint{\mathbb{E}\Big[r_n^u(t)\Big]}{\geq\frac{\mathbb{E}(L_m)\times q_n}{d_m\times \epsilon_t^u}, \; \forall n\in\mathcal{N}, m\in\mathcal{M}, \label{URLLC_time_Const2}}
		\addConstraint{z_{kb}(t)}{\in\{0, 1, \dots, M\}, \;\; \forall k \in \mathcal{K},\; b\in\mathcal{B} \label{puncturing_const}}.
	\end{maxi!}
	
	The optimization problem \eqref{URLLC_Resource_Allocaiton_Problem}  is a combinatorial optimization problem which is an NP-hard problem for which it is difficult to obtain a closed-form solution. To simplify \eqref{URLLC_Resource_Allocaiton_Problem}, we replace the integer variable $z_{kb}$ by a continuous weighting variable $w_{kb}\in [0,1]$, where $w_{kb}$ is the puncturing weight of the RB $b$ by URLLC traffic, i.e., more resources will be punctured from the RBs with higher weighting values. Therefore, we can approximate the eMBB data rate as
	\begin{equation}
	\tilde{r}_{kb}^e(t)=f_b\Big(1-w_{kb}(t)\Big)\log_2\left(1+\frac{p_{kb}(t)h_{kb}(t)}{\sigma^2}\right).
	\end{equation}
	Then, using the  definition of $z_{kb}(t)$, i.e., the number of punctured mini-slots from the RB \textit{b} of eMBB user \textit{k} at time slot \textit{t}, the URLLC data rate in \eqref{URLLC_data_rate} is modified as
	\begin{equation}
	\tilde{r}_{nb}^u(t)=\frac{f_bw_{kb}(t)}{{N}}\log\left(1+\frac{p_{nb}^u(t)h_{nb}^u(t)}{\sigma^2}\right)
	-\sqrt\frac{D_{nb}^u(t)}{c_{nb}^u(t)}Q^{-1}(\vartheta).
	\end{equation} 
	
	We use the Markov’s Inequality to represent the chance constraint \eqref{URLLC_reliability_Const3} as a linear constraint:
	\begin{equation}
	\textsf{Pr}\left[\sum\limits_{n\in\mathcal{N}}r_{n}^u(t)\leq \zeta L(t)\right]\leq \frac{\zeta \mathbb{E}[L]}{\sum\limits_{n\in\mathcal{N}}r_{n}^u(t)}.
	\end{equation}
	
	Accordingly, the URLLC resource allocation problem can be reformulated as follows: 
	\begin{maxi!}[2]                 % mini! = minimize 
		{\boldsymbol{w}}                               % optimization variable
		{\frac{1}{\mu}\log\mathbb{E}_\textrm{h}\left[\exp\biggl(\mu\sum\limits_{k=1}^{K}\tilde{r}_{k}^e(t)\biggr)\right]}{\label{URLLC_Resource_Allocaiton_Problem_approximated}}{\tilde{\textbf{P}}\textbf{3: }}   
		\addConstraint{\sum\limits_{n\in\mathcal{N}}\tilde{r}_{n}^u(\boldsymbol{w})}\geq \frac{\zeta \mathbb{E}[L]}{\Theta^*},{\label{URLLC_reliability_Const3_approximated}}
		%	\addConstraint{\mathbb{E}\Big[r_n^u(t)\Big]}{\geq\frac{\mathbb{E}(L_m)\times q_n}{d_m\times \epsilon_t^u}, \; \forall n\in\mathcal{N}, m\in\mathcal{M}, \label{URLLC_time_Const2}}
		\addConstraint{0\leq w_{kb}}{\leq 1, \;\; \forall k \in \mathcal{K},\; b\in\mathcal{B} \label{const_2_urllc}}.
	\end{maxi!}
	
	\begin{lemma}
		For a given $\tilde{\boldsymbol{x}}$ and $\boldsymbol{p}$, \eqref{URLLC_Resource_Allocaiton_Problem_approximated} is a convex optimization problem.
	\end{lemma}
	\begin{proof}
		It is clear that $r_k^e(t), \;\forall k\in \mathcal{K}$ is a linear with respect to $w_k$ and combining it with $\exp$ and $\log$ functions gives a concave function, for all $0\leq w_{kb}\leq 1$. Furthermore, constraints \eqref{URLLC_reliability_Const3_approximated}, and \eqref{const_2_urllc} are linear constraints with respect to $\boldsymbol{w}$. Maximizing a concave objective function with linear constraints is a convex optimization problem. This proves the convexity of \eqref{URLLC_Resource_Allocaiton_Problem_approximated}.
	\end{proof}
	
	We obtain an equivalent relation of the number of punctured mini-slots $z_{kb}$ with parameters $M$ and $w_{kb}$, i.e., $\boldsymbol{z}^{(i+1)}= \floor*{M\times \boldsymbol{w}^{(i+1)}}$, which ensures constraint \eqref{puncturing_const}.
	
	Next, we analyze the quality of rounding technique by measuring the integrality gap which measures the ratio between the value of $\mathcal{G}(\boldsymbol{\tilde{x}}^*, \boldsymbol{p}^*, \boldsymbol{z}^*)+\alpha\Delta$ achieved by feasible rounded solution and the value of $\mathcal{G}(\tilde{\boldsymbol{x}}^*, \boldsymbol{p}^*, \boldsymbol{z}^*)$ achieved by the relaxed solution. According to the definition and proof of integrality gap in \cite{rounding_1}, we have the following definition and remark:
		\begin{definition}
		Given $\mathcal{G}(\tilde{\boldsymbol{x}}^*, \boldsymbol{p}^*, \boldsymbol{z}^*)$ and its rounded problem $\mathcal{G}(\boldsymbol{\tilde{x}}^*, \boldsymbol{p}^*, \boldsymbol{z}^*)+\alpha\Delta$, the integrality gap is defined by:
		\begin{equation}
		    \varrho=\max\limits_{\boldsymbol{\tilde{x}}}\frac{\mathcal{G}(\boldsymbol{\tilde{x}}^*, \boldsymbol{p}^*, \boldsymbol{z}^*)+\alpha\Delta}{\mathcal{G}(\tilde{\boldsymbol{x}}^*, \boldsymbol{p}^*, \boldsymbol{z}^*)},
		\end{equation}
		where the solution of $\mathcal{G}(\tilde{\boldsymbol{x}}^*, \boldsymbol{p}^*, \boldsymbol{z}^*)$ is obtained through relaxation of $\boldsymbol{x}$, while the solution $\mathcal{G}(\boldsymbol{\tilde{x}}^*, \boldsymbol{p}^*, \boldsymbol{z}^*)+\alpha\Delta$ is obtained after rounding the relaxed variable. We consider that the best rounding is achieved when $\varrho$ $(\varrho\leq1)$ is closer to $1$.
		\end{definition}
		\begin{remark} Given $\mathcal{G}(\tilde{\boldsymbol{x}}^*, \boldsymbol{p}^*, \boldsymbol{z}^*)$ whose instances form a convex set, for every relaxation, the oblivious rounding scheme defined as $\mathcal{G}(\boldsymbol{x}^*, \boldsymbol{p}^*, \boldsymbol{z}^*)$ is individually tight \cite{rounding_1}.
		\end{remark}
	
\begin{algorithm}[t!]
                                	\centering
                                		\caption{\strut: DRRA Algorithm for the eMBB/URLLC coexistence Problem}
                                		\label{alg:profit}
                                		\begin{algorithmic}[1]
                                			\STATE{\textbf{Initialization:} Set $i=0$, $\epsilon>0$, and find initial feasible solutions $(\boldsymbol{\tilde{x}}^{(0)}, \boldsymbol{p}^{(0)}, \boldsymbol{w}^{(0)})$;}
                                			\STATE{Decompose \textbf{P} into \textbf{P1}, \textbf{P2}, and \textbf{P3}};
                                			\STATE{Relax and transform, respectively, the integer variables in \textbf{P1} and \textbf{P3} to their continuous form $\tilde{\textbf{P}}\textbf{1}$ and $\tilde{\textbf{P}}\textbf{3}$};
                                			\REPEAT
                                			\STATE{Compute $\boldsymbol{\tilde{x}}^{(i+1)}$ from $\tilde{\textbf{P}}\textbf{1}$ at given $\boldsymbol{p}^i$, and $\boldsymbol{z}^i$};
                                			\STATE{Compute $\boldsymbol{p}^{(i+1)}$ from \textbf{P2} at given $\boldsymbol{\tilde{x}}^{(i+1)}$, and $\boldsymbol{z}^i$};
                                			\STATE{Compute $\boldsymbol{w}^{(i+1)}$ from $\tilde{\textbf{P}}\textbf{3}$ at given $\boldsymbol{\tilde{x}}^{(i+1)}$, and $\boldsymbol{p}^{(i+1)}$, and set $\boldsymbol{z}^{(i+1)}=\floor*{M\times \boldsymbol{w}^{(i+1)}}$};
                                			\STATE{$i = i + 1$}; 
                                			\UNTIL{ $\parallel\frac{\mathcal{G}^{(i)}-\mathcal{G}^{(i+1)}}{\mathcal{G}^{(i)}} \parallel \ \leq \epsilon$};
                                			\STATE{Generate a binary solution $\boldsymbol{x}^*$ from $\tilde{x}^{(i+1)}$ by using the rounding equation in (14) and solving (15)};
                                % 			\STATE{\textcolor{blue}{Calculate $\varrho$, if $\varrho\leq 1$, consider $\boldsymbol{x^*}=\tilde{x}^{(i+1)}$, $\boldsymbol{p^*}=\boldsymbol{p}^{(i+1)}$ and $\boldsymbol{z^*}=\boldsymbol{z}^{(i+1)}$ as a solution}};
                                			\STATE{Then, set $\big(\boldsymbol{x}^{*},\boldsymbol{p}^{*}=\boldsymbol{p}^{(i+1)} , \boldsymbol{z}^{*}=\boldsymbol{z}^{(i+1)}\big)$ as the desired solution}.
            		\end{algorithmic}
                                		\label{Algorithm}
                    	\end{algorithm}
Algorithm 1 starts by initializing $i=0$, setting $\epsilon$ to a small positive number, and finding initial feasible points ($\boldsymbol{\tilde{x}}^{(0)}, \boldsymbol{p}^{(0)}, \boldsymbol{z}^{(0)}$). Then, the algorithm starts an iterative process. At each iteration $i+1$, the solution is updated by solving $\tilde{\textbf{P}}\textbf{1}$, \textbf{P2}, and $\tilde{\textbf{P}}\textbf{3}$ until achieving $\parallel\frac{\mathcal{G}^{(i)}-\mathcal{G}^{(i+1)}}{\mathcal{G}^{i}}\parallel\leq \epsilon$. Next, Algorithm 1 generates a binary solution for $\boldsymbol{\tilde{x}}^{(i+1)}$ using the rounding technique in (14) and solving (15), where the best rounding is achieved when $\varrho\to 1$.

The value of $\epsilon$ is selected to guarantee an $\epsilon$-optimal solution such that $\parallel\frac{\mathcal{G}^{(i)}-\mathcal{G}^{(i+1)}}{\mathcal{G}^{i}}\parallel\leq \epsilon$ is satisfied. An $\epsilon$-optimal solution is defined as $(\boldsymbol{x}^*, \boldsymbol{p}^*, \boldsymbol{z}^*)\in\{\boldsymbol{x}, \boldsymbol{p}, \boldsymbol{z}|\boldsymbol{x}\in\mathcal{X}, \boldsymbol{p}\in\mathcal{P}, \boldsymbol{z}\in\mathcal{Z}, \mathcal{G}(\boldsymbol{x}, \boldsymbol{p}, \boldsymbol{z})-\mathcal{G}(\bar{\boldsymbol{x}}, \bar{\boldsymbol{p}}, \bar{\boldsymbol{z}})\}$, where $\mathcal{G}(\bar{\boldsymbol{x}}, \bar{\boldsymbol{p}}, \bar{\boldsymbol{z}})$ is globally optimal \cite{hong2015unified}. According to the convergence analysis in Appendix A, the algorithm converges sub-linearly in the order of $1/\epsilon$.        
% 	\vspace{-0.5cm}
	\section{Intelligent URLLC Scheduling: Deep Reinforcement Learning Based Approach}
% 		\vspace{-0.4cm}
	In the previous section, we have proposed a DRRA algorithm to  solve the eMBB resource allocation problem and find an approximate solution for the URLLC scheduling problem. The URLLC scheduling obtained by the DRRA algorithm may violate the URLLC reliability constraint at the worst case conditions due to the relaxation applied to the probability constraint. In practice, URLLC traffic is random and sporadic; thus, it is necessary to dynamically and intelligently allocate resources to the URLLC traffic by interacting with the environment. Therefore, we propose a DRL-based algorithm to tackle the dynamic URLLC traffic and channel variations. In this algorithm, the URLLC reliability constraint is dynamically verified and the system parameters are adjusted as per URLLC requirements. Going further, we leverage the URLLC scheduling results obtained by the DRRA algorithm to learn the proposed DRL-based algorithm at the initial start to improve its convergence time. Hence, combining the advantages of the optimization-based algorithm (DRRA) and the DRL-based algorithm compound in a reliable and efficient resource allocation mechanism. 
	
	Generally, a reinforcement learning model is defined by its action space $\mathcal{A}$, state space $\mathcal{S}$, and reward $R(t)$. The algorithm takes an action $\boldsymbol{a}(t)\in\mathcal{A}$ at each state $\boldsymbol{s}(t)\in\mathcal{S}$ and receives the reward $R(t)$. 
	%\textbf{Agents:} each active gNB-eMBB link represents an agent.\\
	\subsubsection{State space} We consider the state space with the tuples defining the state of each eMBB user, i.e., the allocated RBs, transmission power, and channel variations, and URLLC traffic states, i.e., number of arrived URLLC packets and channel variations, at each decision epoch (time slot). Therefore, the state at time slot $t$ can be defined as $\boldsymbol{s}(t)=\{\boldsymbol{x}(t), \boldsymbol{p}(t), \boldsymbol{h^e}(t), L(t), \boldsymbol{h^u(t)}\}$. In order to reduce the state space dimensions, we define $\hat{r}^e_{k}(t)$ as the data rate of eMBB user $k$ without puncturing:
	\begin{equation}
	\hat{r}_k^e(t)=\sum\limits_{b\in\mathcal{B}}x_{kb}(t)f_b\log_2\bigg(1+\frac{p_{kb}(t)h_{kb}(t)}{\sigma^2}\bigg),
	\end{equation}
	which depends on the allocated RBs, allocated power, and channel state. Therefore, the state space at time slot $t$ can be reduced to $\boldsymbol{s}(t)=\{\boldsymbol{\hat{r}^e}(t), \boldsymbol{h^u}(t), L(t)\}$.
	\subsubsection{Action space} The action space is defined as the number of punctured mini-slots of each RB, $\boldsymbol{a}(t)=\{ z_{kb}(t), \; \forall k\in\mathcal{K}, b\in\mathcal{B}\}$, which is a $B\times M$ puncturing matrix. 
	\subsubsection{Reward} Considering the requirements of eMBB and URLLC services, we formulate the reward function as follows:
	\begin{equation}
	\begin{split}
	R(\boldsymbol{a}(t), \boldsymbol{s}(t))=g(t)
	+\phi(t)\mathbb{E}\bigg[\sum\limits_{n=1}^{N}r_n^u(t)-\zeta L(t)\bigg],
	\end{split}
	\label{reward_fucntion}
	\end{equation}
	where $\phi(t)$ is a time-varying weight that ensures the URLLC reliability over time slots as the network states change dynamically. We define $\phi(t)$ as follows:
	\begin{equation}
	\phi(t+1)=\max\left\{\phi(t)+\Theta(t)-\Theta^{\mathsf{max}},  0\right\}, 
	\label{urllc_weight_updating}
	\end{equation}
	where $\Theta(t)$ is the estimated outage probability at time slot $t$ defined in (7) which can be obtained using an empirical measurement of the number of time slots (in the last $T$ slots) where $\sum_{n\in\mathcal{N}}r_n^u(t)\leq\zeta L(t)$ over $T$. 
	
	The agent aims to choose a policy $\pi(\boldsymbol{a, s})=\{\pi_{b}^m, \; \forall b\in\mathcal{B}, m\in\mathcal{M}\}$, where $\pi_{b}^m$ is the probability of puncturing $m$ mini-slots from the RB $b$ given the network state $\boldsymbol{s}(t)$. Specifically,  the agent observes the network state $\boldsymbol{s}(t)$ and makes a decision on the punctured resources from each RB based on its learned policy strategy. After that, the agent calculates the immediate reward $R(t)$ from \eqref{reward_fucntion} based on the selected actions and provides the new network state information to the agent for the current obtained reward. Finally, the agent learns a new policy in the next decision epoch according to the feedback.
	
	Let $Q^{\pi}(\boldsymbol{s},\boldsymbol{a})$ denote the cumulative discounted reward with a given policy $\boldsymbol{\pi}$, defined as
	\begin{equation}
	Q^{\pi}(\boldsymbol{s},\boldsymbol{a})=\mathbb{E}\left[\sum\limits_{t=0}^{\infty}\gamma(t) R\big(\boldsymbol{s}(t), \boldsymbol{a}(t)\big)|s_0=s, \pi\right]. 
	\label{Q_Funciton}
	\end{equation}
	
	The function $Q^{\pi}(\boldsymbol{s},\boldsymbol{a})$ can be calculated using the Bellman equation \cite{TWC_AC_2014}:
	\begin{equation}
	Q^{\pi}(\boldsymbol{s},\boldsymbol{a})=\mathbb{E}\left[R\big(\boldsymbol{s}(t), \boldsymbol{a}(t)\big)+Q^{\pi}\big(\boldsymbol{s}(t+1), \boldsymbol{a}(t+1)\big)\right].
	\label{Bellman_Update}
	\end{equation}
	
	Let $J(\pi)$ be the network objective reward value, which is defined as \cite{TWC_AC_2014}:
	\begin{equation}
	J(\pi)=\mathbb{E}\Big[Q^{\pi}(\boldsymbol{s},\boldsymbol{a})\Big]=\int_S\int_A\pi(\boldsymbol{s},\boldsymbol{a})Q^{\pi}(\boldsymbol{s},\boldsymbol{a})dads.
	\label{J_function}
	\end{equation}
	
	The objective is to find the policy that maximizes $J(\pi)$. We observe in \eqref{J_function} that it is possible to optimize the policy $\pi$ using different techniques such as the Q-learning, and policy gradient techniques. However, applying the Q-learning method may fail to find the optimal policy in real-time as the learning rate of the Q-function is slow \cite{MIT_RL, AC_2018_TMC}. The policy gradient can provide a good policy with a faster convergence rate than Q-learning. Therefore, we propose a policy gradient based actor-critic learning (PGACL) algorithm to learn policies by combining the policy learning and value learning with a good convergence rate. The PGACL learning has the ability to optimize the policy with a fast convergence rate and  a low computational cost by leveraging the gradient method.
% 	\vspace{-0.7cm}
	\subsection{PGACL algorithm for URLLC scheduling}
% 		\vspace{-0.5cm}
	The PGACL consists of two main parts, namely the actor and the critic. The actor part controls the policy based on the network state, while the critic part evaluates the selected policy by the reward function as shown in Fig. \ref{AC_Fig}.
		\begin{figure}
		\centering
		\includegraphics[width=0.3\linewidth]{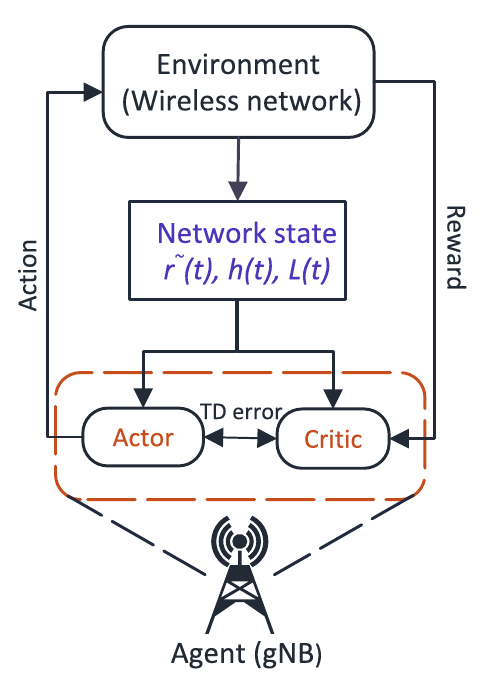}
		\caption{The actor-critic based learning for URLLC scheduling problem.}
		\label{AC_Fig}	
	\end{figure}
		\subsubsection{The actor part} The actor updates the policies based on the policy gradient method. The policy is initially built based on a parameter vector $\boldsymbol{\theta}$ as $\pi_{\boldsymbol{\theta}}(\boldsymbol{s}, \boldsymbol{a})=Pr(\boldsymbol{a}|\boldsymbol{s},\boldsymbol{\theta})$. Here, the gradient of the objective function in \eqref{J_function} with respect to $\boldsymbol{\theta}$ is as follows:
	\begin{equation}
	\nabla_{\theta}J(\pi_{\boldsymbol{\theta}})=\int_{S}\int_A\nabla_{\pi_{\theta_k}}Q^{\pi_{\theta}}(\boldsymbol{s}, \boldsymbol{a})dads.
	\label{gradient_objective_funciton}
	\end{equation}
	
	The parameterized policy $\pi_{\theta}(\boldsymbol{s,a})$ is defined by the Gibbs distribution as follows \cite{TWC_AC_2014}:
	\begin{equation}
	\pi_{\boldsymbol{\theta}}(\boldsymbol{s}, \boldsymbol{a})=\frac{\exp(\theta\Phi(\boldsymbol{s},\boldsymbol{a}))}{\sum_{a'\in \mathcal{A}}\exp(\theta\Phi(\boldsymbol{s},\boldsymbol{a'}))},
	\end{equation}
	where $\Phi(\boldsymbol{s, a})$ is the feature vector.
	
	Finally, the vector $\boldsymbol{\theta}$ is updated using the gradient function in \eqref{gradient_objective_funciton} as follows:
	\begin{equation}
	\boldsymbol{\theta}(t+1)=\boldsymbol{\theta}(t)+\rho_a\nabla_{\theta}J(\pi_{\boldsymbol{\theta}}),
	\end{equation}
	where $\rho_a$ is the learning rate of the actor. 
	\subsubsection{The critic part} The objective of the critic part is to evaluate the policy that the learning algorithm searches. The function estimator is used to approximate the value function as Bellman equation fails to compute the $Q^{\pi}(\boldsymbol{s},\boldsymbol{a})$ function for the infinite states \cite{MIT_RL}. Specifically, the linear function estimator is applied to evaluate the value function. Hence, the approximated value function is given as 	
	\begin{equation}
	V(\boldsymbol{s},\boldsymbol{a})=\boldsymbol{v}^{T}\boldsymbol{\varphi}(\boldsymbol{s},\boldsymbol{a})=\sum\limits_{i\in\mathcal{S}}v_i\varphi_i(\boldsymbol{s},\boldsymbol{a}),
	\label{value_function}
	\end{equation}
	where $\boldsymbol{\varphi}=[\varphi_1(\boldsymbol{s},\boldsymbol{a}), \dots, \varphi_S(\boldsymbol{s},\boldsymbol{a})]^T$ denotes the basis function vector, $\boldsymbol{v}(\boldsymbol{s}, \boldsymbol{a}) = (v_1, \dots, v_S)^T$ is a weight parameter vector. To compute the error between the estimated and real values, the critic uses the Temporal-Difference (TD) method, defined as
	\begin{equation}
	\delta(t)=R(t+1)+\gamma V\big(\boldsymbol{s}(t+1), \boldsymbol{a}(t+1)\big)-V\big(\boldsymbol{s}(t),\boldsymbol{a}(t)\big).
	\label{TD}
	\end{equation}
	
	The weights parameter vector $\boldsymbol{v}(\boldsymbol{s}, \boldsymbol{a})$ is updated by the gradient descent method using the linear function estimator in \eqref{value_function} as follows: 
	\begin{equation}
	\begin{split}
	\boldsymbol{v}(\boldsymbol{s}(t+1), \boldsymbol{a}(t+1))&=\boldsymbol{v}(\boldsymbol{s}(t), \boldsymbol{a}(t))+\rho_c\delta_k(t)\nabla_v V(\boldsymbol{s},\boldsymbol{a})
	\\&=\boldsymbol{v}(\boldsymbol{s}(t), \boldsymbol{a}(t))+\rho_c\delta(t)\boldsymbol{\varphi}(\boldsymbol{s},\boldsymbol{a}),
	\end{split}
	\label{update_value}
	\end{equation}
	where $\rho_c$ is the critic learning rate. Finally, the critic updates the value function in \eqref{value_function} based on value of $\boldsymbol{v}(\boldsymbol{s}, \boldsymbol{a})$ in \eqref{update_value}. 	
	
		\begin{figure*}
		\centering
		\includegraphics[width=0.75\linewidth]{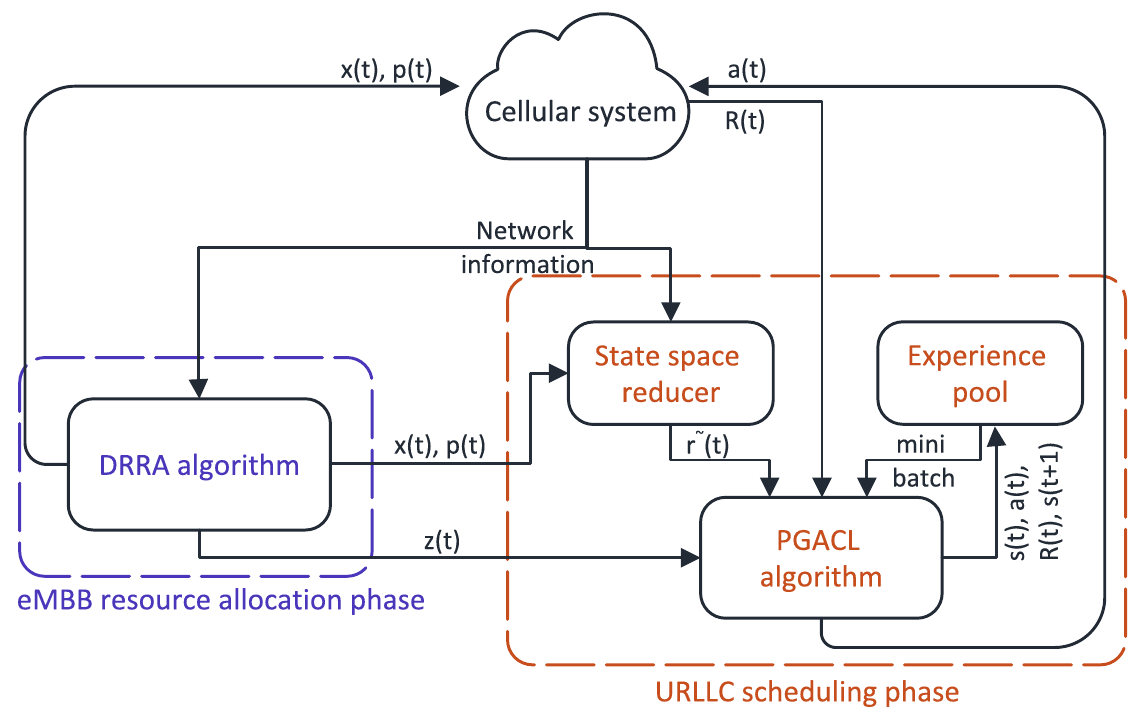}
		\caption{Block diagram of the proposed DRRA-PGACL framework.}
		\label{Block_Diagram}	
	\end{figure*}
	The block diagram of the proposed DRRA-PGACL framework is shown in Fig. \ref{Block_Diagram}. First, the gNB allocates resources to eMBB users based on the optimal results obtained by the DRRA algorithm and forwards it, in addition to the current network state, to the PGACL algorithm. The experience pool of the proposed PGACL algorithm is initialized according to the current optimal solution obtained by the DRRA algorithm. Moreover, the URLLC reliability weight $\phi$ is initialized randomly. Then, the PGACL algorithm selects an action according to the current policy. During the first $\hat{T}$ learning steps, the PGACL algorithm replaces the selected action by $\boldsymbol{z^*}(t)$ obtained by the DRRA algorithm. Next, the PGACL algorithm executes the selected action, observes the immediate reward $R(t)$ and next state $\boldsymbol{s}(t+1)$, and stores the experience tuple $\{\boldsymbol{s}(t), \boldsymbol{a}(t), R(t), \boldsymbol{s}(t+1)\}$ in the experience pool. The network is trained by sampling random tuples from the experience pool. Finally, the value of $\phi(t)$ is updated according to \eqref{urllc_weight_updating}. In the next section, we have detailed simulations to show the convergence time and performance of the proposed algorithms.  
% 	\vspace{-0.5cm}
	\section{Performance Evaluation}
		\begin{table}[t!]
% 	\captionsetup{labelfont={color=blue},font={color=blue}}
		\caption{Simulation Parameters}
		\centering 
% 		\fboxrule=0pt
		\begin{tabular}{|l|l|}
			\hline 
			\textbf{Parameter}  & \textbf{Value} \\
			\hline
		    Cell radius & $300 m$\\
		    		    \hline
		    TTI length & $1$ ms\\
		    		    \hline
		    Time-frame length & $10$ ms\\
		    		    \hline
		    Number of sTTIs in each TTI & 7\\
		    		    \hline
		    Number of sub-carriers per RB & 12\\
		    		    \hline
		    Number of OFDM symbols in each time slot & 14\\
		    		    \hline
		    Number of OFDM symbols in each mini-slot & {2}\\
		    		    \hline
		    Sub-carrier-spacing & $15$ kHz\\
		    		    \hline
		    Total system bandwidth & $20$ MHz\\
		    		    \hline
		    URLLC packet size & $32$\\
		    		    \hline
		    URLLC traffic model & Poisson process with rate $\lambda_u$\\
		    		    \hline
		    eMBB traffic model & Full-buffered\\
		    		    \hline
		    Actor learning rate &$10^{-5}$\\
		    		    \hline
		    Critic learning rate & $10^{-3}$\\
		    \hline
		    Mini-batch size & $32$\\
		    \hline
		    \end{tabular}
		    \end{table}
% 		\vspace{-0.4cm}
	In this section, we validate the efficacy of our proposed algorithms via comprehensive experimental analysis. We assess the performance of proposed solution approach for different parameter settings. To that end, we compare our results with the following state-of-the-art schedulers: \emph{1) MAT \cite{Bairagi2019}:} one-sided matching game is used to take-over the eMBB users resources for supporting URLLC traffic, \emph{2) LMCS \cite{baseline_2}:} URLLC traffic is scheduled by dropping eMBB users with lowest modulation coding scheme (MCS), \emph{3) Sum-Rate:} puncturing strategy is adopted to maximize the average sum-rate of eMBB users, and \emph{4) Sum-Log:} wireless resources are allocated so as to maximize the sum-log of eMBB users data rate, i.e., proportional fair allocation.
% 	\vspace{-0.7cm}
	\subsection{Simulation Environment}
% 		\vspace{-0.5cm}
	We consider a wireless network, where one gNB is deployed at the center of the coverage area with a radius of $200$ m. A number of eMBB and URLLC users are distributed randomly within the coverage area. The duration of a time slot is set to $1$ ms and each time slot is further divided into $7$ equally spaced mini-slots. Each RB is composed of $12$ subcarriers with $14$ OFDM symbols and each subcarrier has a subcarrier-spacing of $15$ kHz. Thus, the bandwidth of each RB is $180$ kHz and each mini-slot consists of $2$ OFDM symbols. Moreover, the total system bandwidth is set to $20$ MHz. We consider the arrival of URLLC packets in each mini-slot follows Poisson process with rate $\lambda_u$ and the size of each packet is $32$ bytes. The complete list of simulation parameters is given in Table II.
% 	\vspace{-0.7cm}
	\subsection{Performance analysis of the DRRA algorithm}
		\begin{figure}[t!]
    \centering
    \begin{minipage}[b]{0.45\linewidth}
    \includegraphics[width=1\linewidth]{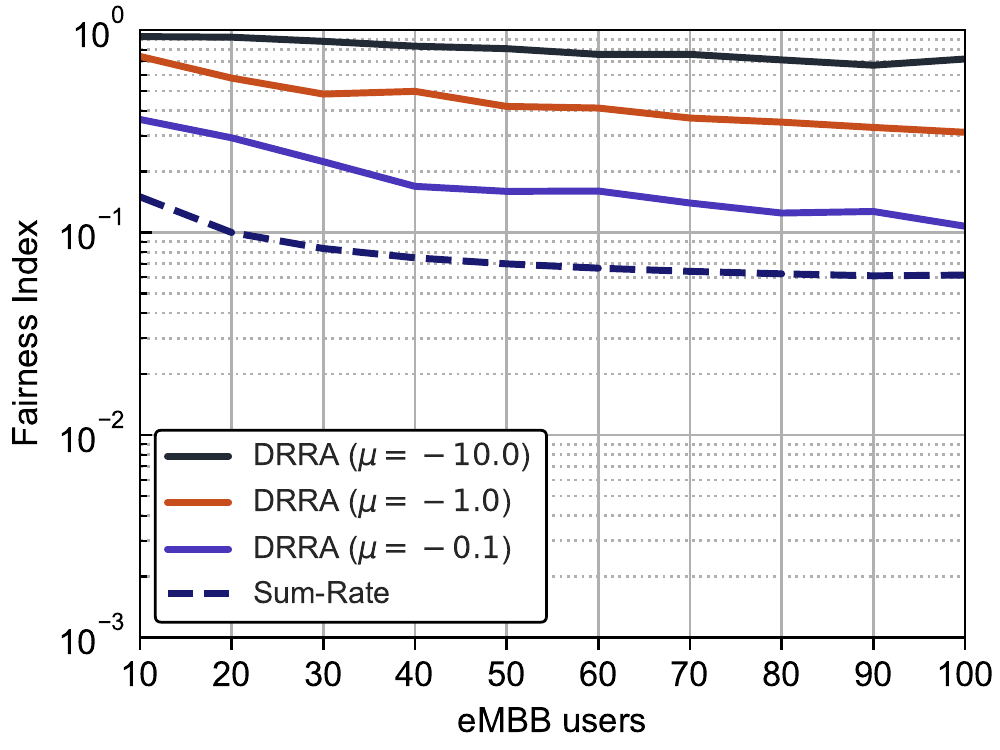}
	\caption{\small The Jain's fairness among eMBB users.}
	\label{eMBB_Fairness}
    \end{minipage}
    \hfill
    \begin{minipage}[b]{0.45\linewidth}
    \includegraphics[width=1\linewidth]{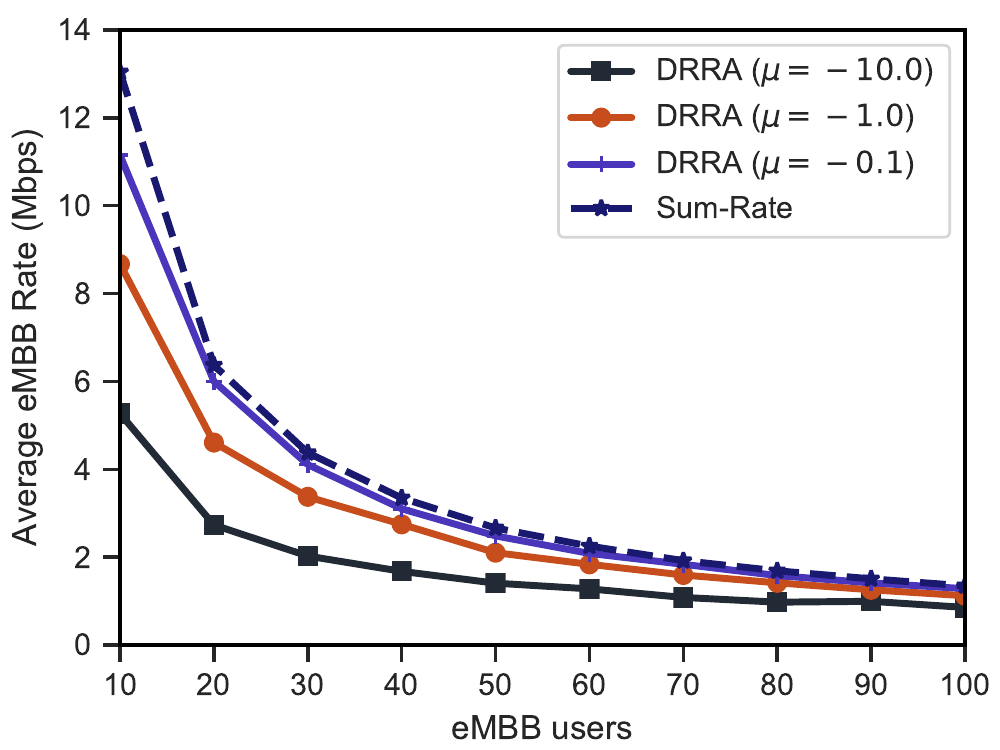}
	\caption{\small Average per user eMBB data rate.}
	\label{avg_eMBB_rate}
    \end{minipage}
    \end{figure}
	We first study the performance of the proposed DRRA algorithm, in the eMBB resource allocation phase, for different parameter configurations and compare it with the \emph{Sum-Rate} baseline, where the objective is to maximize the sum-rate of all eMBB users. Specifically, we study the fairness among eMBB users in Fig. \ref{eMBB_Fairness} for the proposed DRRA algorithm under different settings and compare it to the Sum-Rate approach. The fairness among eMBB users is calculated based on the Jain's Fairness index. As shown in Fig. \ref{eMBB_Fairness}, increasing the value of $\mu$ negatively leads to more risk-averse and hence reduces the variance of eMBB users' data rate. We can see from Fig. \ref{eMBB_Fairness} that $\mu=-10$ ensures fairness by around $90\%$. However, the fairness index is breaking down when we set the value of $\mu$ to $-0.1$ as the algorithm nears to the risk-neutral case where the algorithm maximizes the average sum data rate giving results closer to that of the \emph{Sum-Rate} approach. Furthermore, the \emph{Sum-Rate} approach gives the worst fairness as its objective is to maximize the average sum data rate only without considering its variance, i.e., it allocates more resource to users at good channel states. In Fig. \ref{avg_eMBB_rate}, we study the average per user data rate for different values of $\mu$. The \emph{Sum-Rate} approach provides the highest data rate as its objective is to maximize the average data rate without considering the QoS requirements of each eMBB user resulting in unreliable transmission. However, the proposed DRRA algorithm with lower values of $\mu$ gives lower average data rate as the algorithm gives higher priority to the variance and hence allocates more resources to the users at bad channel states ensuring more reliable transmission. Moreover, setting $\mu$ to high values gives results comparable  to the \emph{Sum-Rate} approach. 
	\begin{figure}[t!]
		\centering
		\begin{minipage}[b]{0.45\linewidth}
			\includegraphics[width=1\linewidth]{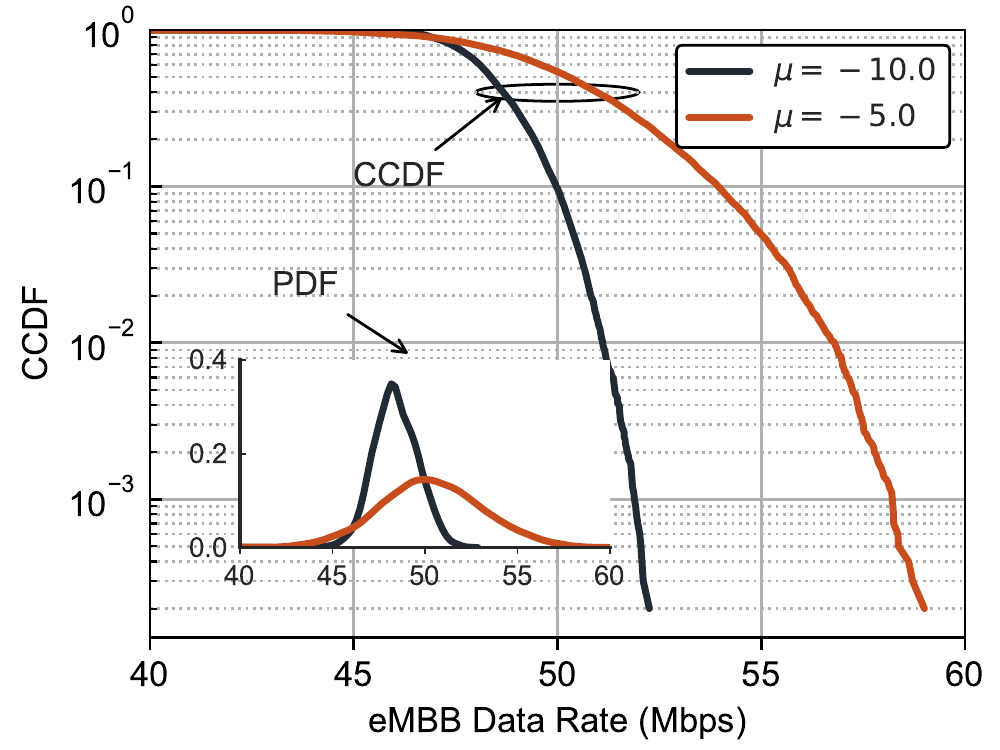}
			\caption{\small CCDF and PDF of the sum eMBB data rate.}
			\label{eMBB_rate_variance}
		\end{minipage}
		\hfill
		\begin{minipage}[b]{0.45\linewidth}
			\includegraphics[width=1\linewidth, height=2.2in]{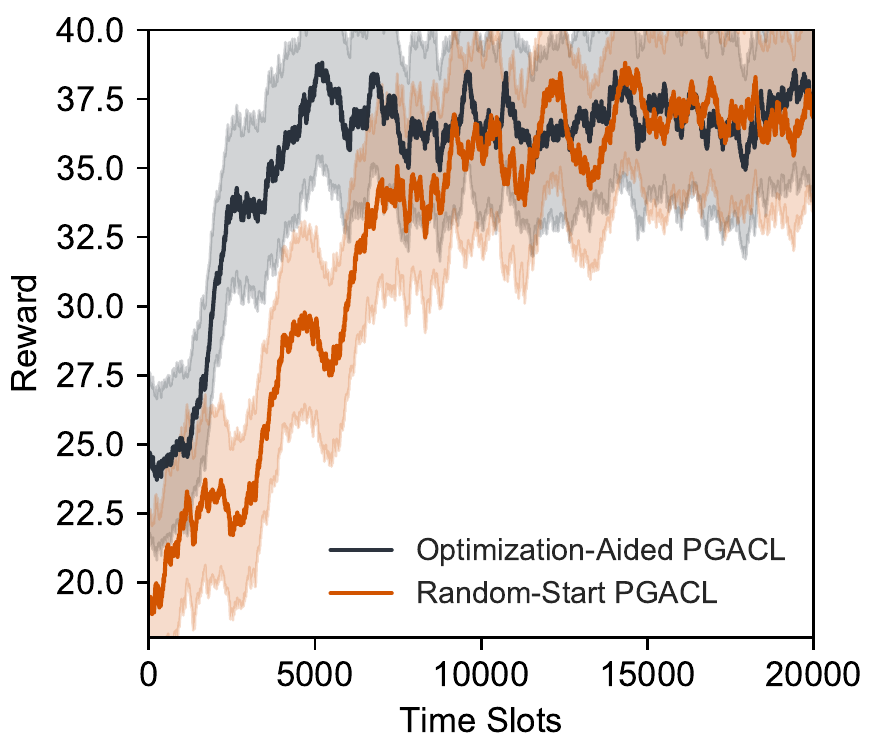}
			\caption{\small Convergence of the reward value over time slots.}
			\label{reward_convergence}
		\end{minipage}
	\end{figure}
  
	In Fig. \ref{eMBB_rate_variance}, we plot the complementary cumulative distribution function (CCDF) and the probability density function (PDF) of the eMBB data rate calculated over time for different values of $\mu$. Setting $\mu$ to higher negative values degrades the eMBB sum data rate while reducing its variance which leads to more stable and reliable eMBB transmissions over time. As shown in Fig. \ref{eMBB_rate_variance}, the average eMBB sum data rate is around $50$ Mbps and it varies from $40$ Mbps to $60$ Mbps when $\mu=-5.0$. However, setting $\mu=-10.0$ gives data rate between $45$ Mbps to $52$ Mbps resulting in a stable eMBB transmission.
% 	\vspace{-0.7cm}
	\subsection{Convergence analysis of the PGACL algorithm}
	We study the convergence of the proposed \emph{optimization-aided PGACL} algorithm, i.e., pre-trained using the results obtained by the DRRA algorithm, and compare it with the \emph{Random-Start PGACL} approach, where the PGACL algorithm is initialized with random data. Specifically, Fig. \ref{reward_convergence} shows the convergence of the reward function over time. As shown in Fig. \ref{reward_convergence}, the  algorithm incurs a worse performance at the beginning when initializing it with a random data and improves over time. On the other hand, the proposed \emph{optimization-aided PGACL} algorithm leverages the results of the DRRA algorithm for training during the first time slots enabling fast convergence and hence achieving better response to the dynamic environment.
	\subsection{URLLC reliability analysis}
	    \begin{figure}[t!]
    \centering
    \begin{minipage}[b]{0.45\linewidth}
    \includegraphics[width=1\linewidth, height=2.2in]{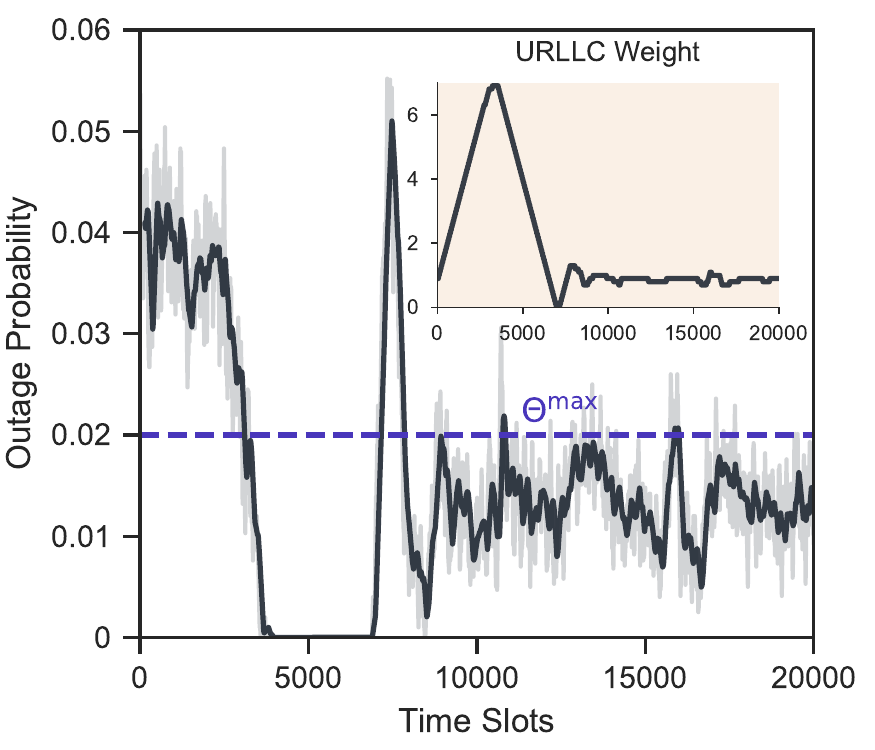}
		\caption{\small The convergence of the outage probability.}
		\label{outage_prob_convergence}
    \end{minipage}
    \hfill
    \begin{minipage}[b]{0.45\linewidth}
    \includegraphics[width=1\linewidth]{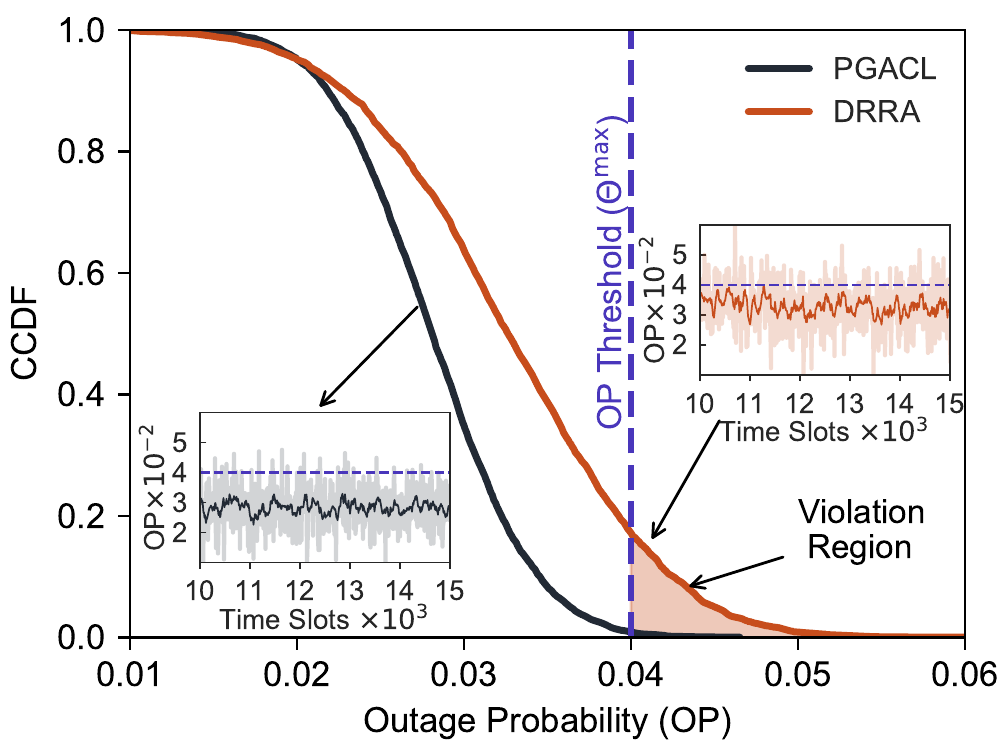}
		\caption{\small CCDF of the URLLC outage probability.}
		\label{URLLC_Reliability}
    \end{minipage}
    \end{figure}
	First, we discuss the convergence of the URLLC outage probability during the learning process in Fig. \ref{outage_prob_convergence}. It is clear that the outage probability converges to a value lower than $\Theta^*$ as the algorithm checks the reliability constraint at each time slot and then updates the value of $\phi(t)$ to ensure the URLLC reliability constraint. Moreover, the updating values of $\phi(t)$  over time slots based on the proposed updating rule in \eqref{urllc_weight_updating} is included in Fig. \ref{outage_prob_convergence}. Next, we discuss the worst case of the URLLC reliability obtained by the DRL-based PGACL algorithm and compare it with that of the optimization-based DRRA algorithm in Fig. \ref{URLLC_Reliability}. We plot the CCDF of the URLLC reliability to emphasize its tail distribution. It is shown that the DRL-based PGACL algorithm minimizes the tail-risk of the URLLC outage probability and ensures its values less than $\Theta^*$ while the optimization based DRRA algorithm fails to capture the worst case violating the URLLC reliability. The DRL-based PGACL algorithm learns the URLLC traffic and channel variations and adjusts the URLLC weight dynamically based on \eqref{urllc_weight_updating}, which leads to more reliable transmissions. However, the optimization-based DRRA algorithm fails to ensure stringent outage probability due to the applied relaxation methods to get a convex form. As shown in Fig. \ref{URLLC_Reliability}, the outage probability obtained by the DRRA algorithm may violate the reliability constraint with a violation probability around $0.18$ when setting $\Theta^*=0.04$ while the PGACL algorithm can ensure stringent reliability. 
    % \setlength{\textfloatsep}{6pt}
% \setlength{\textfloatsep}{6pt}
% 	\vspace{-0.7cm}
	\subsection{Impact of URLLC traffic on eMBB reliability}
		\begin{figure*}[t]
		\centering
		\includegraphics[width=1\linewidth, height=2.0in]{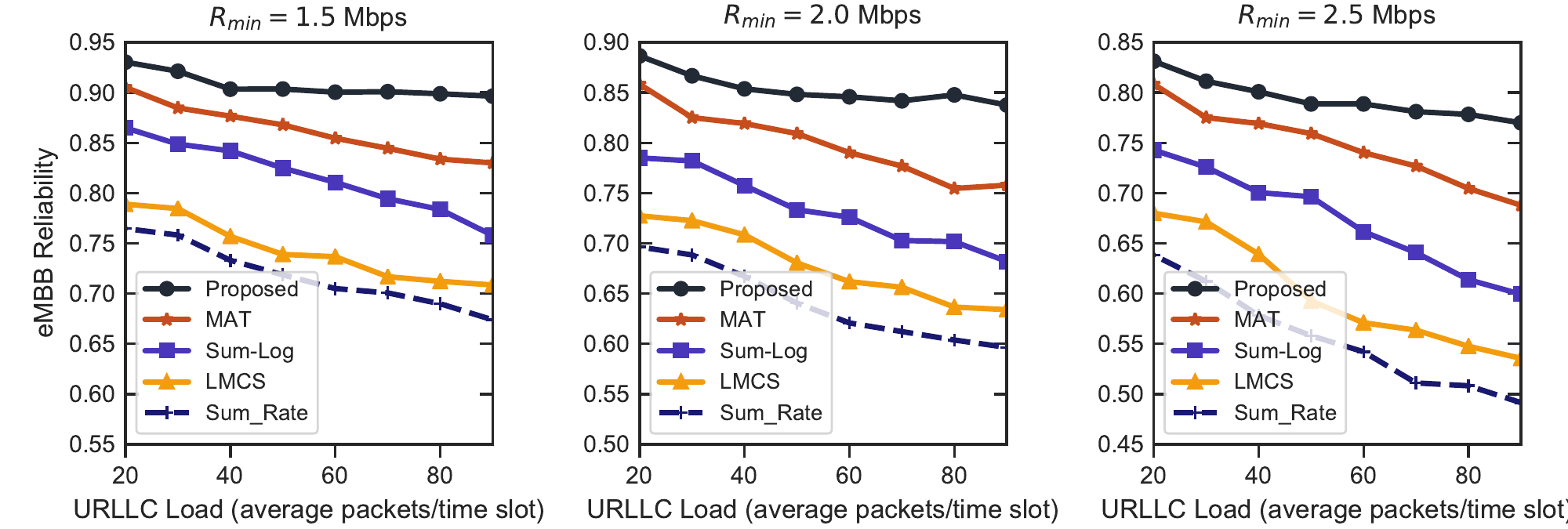}
		\caption{\small eMBB reliability for different $L$ and $R_{\textrm{min}}$.}
		\label{eMBB Reliability}
% 		\vspace{-0.5cm}
	\end{figure*}
	We study the impact of URLLC traffic on the eMBB reliability and compare the results obtained by the proposed risk-averse based approach with the \emph{MAT, LMCS, Sum-Log,} and, \emph{Sum-Rate} baselines. The eMBB reliability is calculated as the number of eMBB users with data rate higher than a target rate $R_{\textrm{min}}$ divided by the total number of eMBB users\footnote{{We quantify the reliability of eMBB users as the proportion of satisfied eMBB users for the given channel conditions and incoming URLLC traffic. Thus, this measurement translates well to capture the eMBB reliability defined by the variance part in the formulated problem \eqref{URLLC_objjective_function}.}}. Fig. \ref{eMBB Reliability} shows that the proposed algorithm guarantees higher reliability as compared to the others baselines. In the Sum-Rate approach, the algorithm tries to maximize the sum data rate of eMBB users by puncturing eMBB users with low data rate. However, protecting eMBB users having higher data rate eventually degrades the reliability of eMBB transmissions. In the \emph{LMCS} approach, URLLC traffic is scheduled on the eMBB users with low MCS. A lower code rate is used in poor channel conditions implementing low-order modulation schemes such as BPSK and QPSK, which are more robust and can tolerate higher levels of interference. Thus, we can notice from Fig. \ref{eMBB Reliability} that \emph{LMCS} gives results close to \emph{Sum-Rate}. However, in the \emph{MAT} approach, a one-sided matching game is applied to schedule URLLC traffic on eMBB users given the objective to maximize the average eMBB data rate while protecting the QoS requirements of eMBB users. Thus, we can notice that \emph{MAT} gives better eMBB reliability compared to \emph{LMCS} scheduler. Furthermore, the Sum-Log approach distributes URLLC traffic equally among all eMBB users resulting in moderated reliability. However, the proposed risk-averse algorithm considers the variance of eMBB users and allows to protect users at bad channel states by puncturing those at better states, which further enhances eMBB reliability. We can also see that eMBB reliability decreases when increasing $R_{\textrm{min}}$. 
				\begin{figure}[t]
		\centering
		\includegraphics[width=0.45\linewidth]{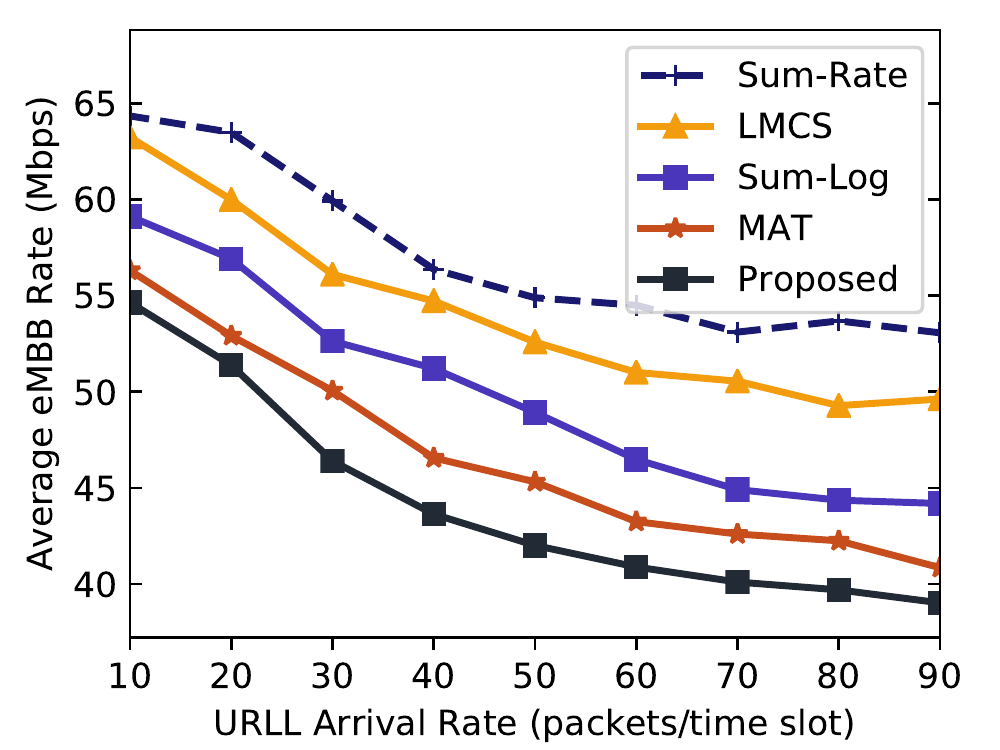}
		\caption{\small Average eMBB data rate for different values of $L$.}
		\label{eMBB_sum_rate} 
	\end{figure}
	
	As shown in Fig. \ref{eMBB Reliability}, the proposed approach keeps the eMBB reliability higher than $90\%$ at $R_{\textrm{min}}=1.5$ Mbps while the Sum-Rate fails to maintain an acceptable reliability, which breaks down to lower than $75\%$. Moreover, the proposed approach provides reliability higher than $80\%$ when increasing $R_{\textrm{min}}$ to $2.5$ Mbps while the reliability obtained by the Sum-Rate breaks down to lower than $60\%$. Furthermore, it is observed that an increase in the URLLC traffic decreases the eMBB reliability as we need to puncture more resources from eMBB users.
% 	\vspace{-0.7cm}
	\subsection{Impact of URLLC traffic on eMBB data rate}
	Finally, we discuss the impact of URLLC traffic on the average eMBB data rate. In doing so, we plot the average eMBB data rate for different URLLC traffic loads and compare the results obtained by the proposed algorithm with other baselines. Fig. \ref{eMBB_sum_rate} shows that increasing URLLC traffic degrades the eMBB data rate. This is because the gNB prioritizes URLLC traffic over eMBB traffic and allocates more resources to satisfy URLLC reliability requirements. Moreover, the \emph{Sum-Rate} approach provides higher average data rate compared to the other approaches as its objective is to maximize the linear summation of eMBB data rate only without considering the eMBB reliability. Moreover, the \emph{LMCS} assigns higher puncturing weight to the eMBB users with lower MCS, and thus resulting in higher average data rate compared to the proposed approach. However, markedly different from these state-of-the-art approaches, the proposed algorithm considers both the average eMBB rate and its variance, and hence achieves a balance between data rate and reliability, as shown in Fig. \ref{eMBB Reliability} and Fig. \ref{eMBB_sum_rate}.
	
	As shown in Fig. \ref{eMBB_sum_rate}, the \emph{Sum-Rate} approach provides an average sum eMBB data rate of $64$ Mbps when the average URLLC load is $10$ (packets/time slot) and decreases to $55$ Mbps when increasing the average URLLC load to $90$ packets/time slot. However, the average sum data rate obtained by the proposed approach varies from $55$ Mbps to $40$ Mbps when increasing the average URLLC load from $10$ to $90$ packets/time slot.
% \setlength{\textfloatsep}{6pt}
% 	\vspace{-0.5cm}
	\section{Conclusion}
% 		\vspace{-0.4cm}
	In this paper, we have studied the coexistence problem of eMBB and URLLC services in 5G networks. We have formulated a risk-sensitive based formulation to improve the reliability of both eMBB and URLLC services. In particular, we have proposed an optimization-aided DRL-based approach that combines the advantages of optimization and learning methods for solving the resource allocation problem. Specifically, resources are allocated to eMBB users at the eMBB resource allocation phase. Moreover, the eMBB resource allocation phase is leveraged to schedule the URLLC traffic at the initial stage and its results are used to learn the DRL-based algorithm to enhance its convergence. In the URLLC scheduling phase, we have proposed a DRL-based learning algorithm in the actor-critic architecture to distribute the URLLC traffic across the ongoing eMBB transmission. Through extensive simulations, we have verified that the proposed algorithms can satisfy the stringent requirements of URLLC while protecting the eMBB reliability.   
% 	\vspace{-0.7cm}
\appendices
% 	\vspace{-0.4cm}
  \section{Convergence Analysis}
  \label{appendixA}
		The proof is based on \cite{xu2013block} which shows the global convergence condition and the asymptotic convergence rate by using the assumption of the Kurdyka-Lojasiewicz property. \par
        
        The general block multi-convex function is in form of
        % \vspace{-0.5cm}
		\begin{equation}
		    \min\limits_{x\in\mathcal{X}} G(\boldsymbol{x}_1,\dots,\boldsymbol{x}_s):=g(\boldsymbol{x}_1,\dots, \boldsymbol{x}_s)+\sum\limits_{k=1}^{s}h_k(\boldsymbol{x}_k)
		\end{equation}
% 		\vspace{-0.4cm}
		where variable $\boldsymbol{x}$ is decomposed into $s$ blocks $\boldsymbol{x}_1,..., \boldsymbol{x}_s$, $g$ is assumed to be a differentiable and block multi-convex function, $h_k$ is the convex function for each block, and a set $\mathcal{X}$ is a block multi-convex set. In specific, the function $g$ is a convex function and the set $\mathcal{X}$ is a convex set of each block $\boldsymbol{x}_k$ while other blocks are fixed. Note that the joint constraint among blocks can be modeled in $\mathcal{X}$, and individual constraints for each block can be modeled as the indicator function $h_k$ for each block $\boldsymbol{x}_k$. The convergence analysis requires the following assumptions
% 		\vspace{-0.4cm}
		\begin{assumption}
		$G$ is continuous in $\textrm{dom}(G)$ and $\inf\limits_{x\in dom(G)}G(x)>-\infty$.
		\end{assumption}
		\begin{assumption}
		A function $\psi(\boldsymbol{x})$ satisfies the Kurdyka-Lojasiewicz (KL) property at point $\bar{\boldsymbol{x}}\in\textrm{dom}(\partial \psi)$ if there exists $\theta\in[0;1)$ such that 
		\begin{equation}
		    \frac{|\psi(\boldsymbol{x}-\psi(\bar{\boldsymbol{x}}))|}{\textrm{dist}(\boldsymbol{0},\partial\psi(\boldsymbol{x}))}
		\end{equation}
		is bounded around $\bar{\boldsymbol{x}}$, where $\textrm{dom}(\partial\psi)=\{\boldsymbol{x}:\partial\psi(\boldsymbol{x})\neq 0\}$ and $\textrm{dist}(\boldsymbol{0},\partial\psi(\boldsymbol{x}))= \min\{\parallel\boldsymbol{y}\parallel: \boldsymbol{y}\in \partial \psi(\boldsymbol{x})\}$.  
		\end{assumption}
% 		\vspace{-0.5cm}
		\begin{assumption}
		The initial point $\boldsymbol{x}_0$ is sufficiently close to the critical point $\bar{\boldsymbol{x}}$, and the value of the function $G(\boldsymbol{x}_n)$ for each iteration $n$ is always greater than the value function at $\bar{\boldsymbol{x}}$ (i.e., $G(\boldsymbol{x}_n)>F(\bar{\boldsymbol{x}}), \; \forall n\geq 0$).
		\end{assumption}
		Therefore, with regard to the Assumptions 1-3, the sequence ${\boldsymbol{x}^n}$ converges globally to the closest critical point $\bar{\boldsymbol{x}}$ (or stationary point) [Theorem 2.8 in \cite{xu2013block}] by proving the bounded of $\sum_{n=N}^{\infty}||x^n-x^{n+1}||$. According to Theorem 2.9 in \cite{xu2013block}, the asymptotic convergence rate depends on the parameter $\theta$.  When $\theta=2/3$, we obtain the sub-linear convergence rate $\mathcal{O}(\frac{1}{n})$.	

	\bibliographystyle{IEEEtran}
	\bibliography{References}

\end{document}